%% file: main-arxiv.tex
\documentclass[a4paper,UKenglish,cleveref, autoref, thm-restate]{lipics-v2021}

\title{The Logical Essence of Compiling With Continuations} 

\author{Jos\'e {Esp\'irito Santo}}{Centre of Mathematics, University of Minho, Portugal \and}{jes@math.uminho.pt}{https://orcid.org/0000-0002-6348-5653}{} 

\author{Filipa Mendes}{Centre of Mathematics, University of Minho, Portugal}{filipasimoesmendes@gmail.com}{}{}

\authorrunning{Jos\'e {Esp\'irito Santo} and Filipa Mendes} 

\Copyright{Jos\'e {Esp\'irito Santo} and Filipa Mendes} 



\category{} 

\relatedversion{} 

\funding{The first author was partially financed by Portuguese Funds through FCT (Fundação para a Ciência e a Tecnologia) within the Projects UIDB/00013/2020 and UIDP/00013/2020}

\nolinenumbers 

\EventLongTitle{8th International Conference on Formal Structures for Computation and Deduction (FSCD 2023)}
\EventShortTitle{FSCD 2023}
\EventAcronym{FSCD}
\EventYear{2023}
\EventDate{July 3--6, 2023}
\EventLocation{Rome, Italy}
\EventLogo{}

\usepackage{booktabs}   
\usepackage{subcaption} 

\usepackage{proof}

\usepackage{amsthm}

\usepackage{bbold}

\usepackage{framed}

\usepackage[all]{xy}

\begin{document}

\maketitle

\input{macros}

\newcommand{\whereproofs}[2]{#1} 

\begin{abstract}
The essence of compiling with continuations is that conversion to continuation-passing style (CPS) is equivalent to a source language transformation converting to administrative normal form (ANF). Taking as source language Moggi's computational lambda-calculus ($\lbc$), we define an alternative to the CPS-translation with target in the sequent calculus LJQ, named \emph{value-filling style} (VFS) translation, and making use of the ability of the sequent calculus to represent contexts formally. The VFS-translation requires no type translation: indeed, double negations are introduced only when encoding the VFS target language in the CPS target language. This optional encoding, when composed with the VFS-translation reconstructs the original CPS-translation. Going back to direct style, the ``essence'' of the VFS-translation is that it reveals a new sublanguage of ANF, the \emph{value-enclosed style} (VES), next to another one, the \emph{continuation-enclosing style} (CES): such an alternative is due to a dilemma in the syntax of $\lbc$, concerning how to expand the application constructor. In the typed scenario, VES and CES correspond to an alternative between two proof systems for call-by-value, LJQ and natural deduction with generalized applications, confirming proof theory as a foundation for intermediate representations.
\end{abstract}

\input{introduction}

\input{background}

\input{LJQ}

\input{VFS}

\input{back-direct-style}

\input{conclusions}
\bibliography{bibrefs}


\appendix
\input{originalLJQ}
\input{kernel-lbc}
\input{proofs}

\end{document}

%% file: macros.tex

\newtheorem{defn}{Definition}
\newtheorem{lem}{Lemma}
\newtheorem{prop}{Proposition}
\newtheorem{thm}{Theorem}
\newtheorem{cor}{Corollary}

\newcommand{\SKIP}[1]{}

\newcommand{\warn}[1]{{\color{red}#1}}

\newcommand{\ovl}[1]{\overline{#1}}
\newcommand{\udl}[1]{\underline{#1}}

\newcommand{\mathsc}[1]{{\scriptstyle{#1}}} 


\newcommand{\lb}{\lambda}

\newcommand{\lt}[3]{\mathsf{let}\,#1:=#2\,\mathsf{in}\,#3}
\newcommand{\ltc}[3]{\mathsf{LET}\,#1:=#2\,\mathsf{in}\,#3}

\newcommand{\cutsymbol}{\mathsf{C}}
\newcommand{\rt}[1]{\uparrow\!#1} 
\newcommand{\li}[4]{#1(#2,#3.#4)}  
\newcommand{\CVW}[3]{\cutsymbol_1(#1,#2.#3)} 
\newcommand{\CVN}[3]{\cutsymbol_2(#1,#2.#3)} 
\newcommand{\CMN}[3]{\cutsymbol_3(#1,#2.#3)} 
\newcommand{\Cut}[3]{\cutsymbol(#1,#2.#3)} 
\newcommand{\Cutc}[3]{\cutsymbol(#1:#2.#3)} 
\newcommand{\Cutv}[3]{\cutsymbol_v(#1,#2.#3)} 
\newcommand{\Cutvc}[3]{\cutsymbol_v(#1:#2.#3)} 
\newcommand{\Cutvfs}[2]{\cutsymbol_v(#1,#2)} 
\newcommand{\Cutvfsc}[2]{\cutsymbol_v(#1:#2)} 
\newcommand{\garg}[3]{(#1,#2.#3)} 

\newcommand{\gapp}[4]{#1(#2,#3.#4)}  


\newcommand{\sub}[3]{[#1/#2]#3}
\newcommand{\lsub}[3]{[#1\backslash#2]#3} 


\newcommand{\E}{\mathbb{E}}  

\newcommand{\CC}{\mathbb{C}}  
\newcommand{\cc}{\mathbb{c}}  

\newcommand{\K}{\mathbb{K}}  

\newcommand{\fhole}[1]{[#1]}
\newcommand{\ehole}{\fhole{\_}}


\newcommand{\Bv}{B_v}
\newcommand{\sigmav}{\sigma_v}
\newcommand{\id}{\eta_{cut}}
\newcommand{\idv}{id_v}

\newcommand{\betav}{\beta_v}
\newcommand{\assoc}{assoc}
\newcommand{\ltv}{let_v}
\newcommand{\ltmn}{let_1}
\newcommand{\ltvn}{let_2}
\newcommand{\etalt}{\eta_{let}}
\newcommand{\etac}{\eta_{c}}
\newcommand{\etact}{\eta_{cut}}
\newcommand{\etak}{\eta_{k}}

\newcommand{\red}{\to}
\newcommand{\redd}{\twoheadrightarrow}
\newcommand{\eq}{=}

\newcommand{\eqc}{=_{\mathsf{C}}}
\newcommand{\conv}{\leftrightarrow}

\newcommand{\imp}{\supset}

\newcommand{\seqv}[3]{#1\to#2:#3} 
\newcommand{\seq}[3]{#1\Rightarrow#2:#3} 
\newcommand{\seqc}[4]{#1|#2\Rightarrow#3:#4} 

\newcommand{\seqC}[3]{#1\vdash_{\Csymb}#2:#3} 
\newcommand{\vdashc}{\vdash_{\mathsf{C}}}

\newcommand{\seqCPS}[3]{#1\vdash_{\mathsf{CPS}}#2:#3}
\newcommand{\ac}{\mathscr{A}}  
\newcommand{\bc}{\mathscr{B}}  

\newcommand{\seqJ}[3]{#1\vdash_{\Jsymb}#2:#3}

\newcommand{\Jsymb}{\mathsf{J}}
\newcommand{\lbjv}{\lb\Jsymb_v}

\newcommand{\Csymb}{\mathsf{C}}
\newcommand{\lbc}{\lb{\Csymb}} 

\newcommand{\LJQ}{LJQ}
\newcommand{\lbq}{\lb Q} 
\newcommand{\lbqo}{\lb LJQ} 
\newcommand{\lbqmo}{\lb LjQ} 

\newcommand{\ANF}{\underline{ANF}}
\newcommand{\LNF}{\underline{LNF}}

\newcommand{\VES}{VES}
\newcommand{\VFS}{VFS}
\newcommand{\RCPS}{\underline{CPS}} 
\newcommand{\CPS}{CPS}              

\newcommand{\cnf}{\mathsc{CNF}}
\newcommand{\ces}{\mathsc{CES}}
\newcommand{\vfs}{\mathsc{VFS}} 
\newcommand{\cps}{\mathsc{CPS}} 


\newcommand{\smp}[1]{#1^{\surd}}
\newcommand{\smpv}[1]{#1^{\surd\!\!\!\surd}} 

\newcommand{\knl}[1]{#1^{\triangledown}}
\newcommand{\knlv}[1]{#1^{\triangledown\!\!\!\triangledown}}  

\newcommand{\cpsv}[1]{#1^{\dagger}}  
\newcommand{\cpsc}[2]{(#1:#2)}       
\newcommand{\cpst}[1]{\overline{#1}} 
\newcommand{\cpsk}[1]{#1^{\star}}    

\newcommand{\uCPS}{\underline{CPS}}  

\newcommand{\vfsv}[1]{#1^{\circ}}  
\newcommand{\vfsc}[2]{(#1;#2)}       
\newcommand{\vfst}[1]{#1^{\bullet}} 

\newcommand{\ngv}[1]{#1^{\sim}}   
\newcommand{\ngt}[1]{#1^{-}}   
\newcommand{\angt}[1]{#1^{\wr}}   
\newcommand{\angc}[1]{#1^{\wr}}   

\newcommand{\posv}[1]{#1^{\times\!\!\!\times}}   
\newcommand{\post}[1]{#1^{+}}                    
\newcommand{\apos}[1]{#1^{\times}}               

\newcommand{\dlv}[1]{#1^{\natural}}
\newcommand{\dlt}[1]{#1^{\sharp}}

\newcommand{\iko}[1]{i#1}

%% file: introduction.tex
\section{Introduction}\label{sec:intro}


The conversion of a program in a source call-by-value language to continuation-passing style (CPS) by an optimizing translation that reduces on the fly the so-called administrative redexes produces programs which can be translated back to direct style, so that the final result, obtained by composing the two stages of translation, is a new program in the source language which can be obtained from the original one by reduction to administrative normal form (ANF) --  a program transformation in the source language \cite{FlanaganSabryDubaFelleisen93,SabryFelleisen93}. This fact has been dubbed the ``essence'' of compiling with continuations and has had a big impact and generated an on-going debate in the theory and practice of compiler design \cite{FlanaganSabryDubaFelleisen2003,Kennedy07,MaurerDownenAriolaJones2017}.

Our starting point is the refinement of that ``essence'', obtained in \cite{SabryWadler97}, in the form of a reflection of the CPS target in the computational $\lb$-calculus \cite{MoggiLFCS88}, the latter playing the role of source language and here denoted $\lbc$ -- see Fig.~\ref{fig:essence}. Then we ask: What is the proof-theoretical meaning of this reflection? What is the logical reading of this reflection in the typed setting? Of course, the CPS-translation has a well-known logical reading as a negative translation, based on the introduction of double negations, capable of translating a classical source calculus with control operators \cite{MeyerWand1985,GriffinPOPL90,SorensenUrzyczyn2006}. But it is not clear how this reading is articulated with the reflection in Fig.\ref{fig:essence}, which provides a \emph{decomposition} of the CPS-translation as the reduction to ANF followed by a ``kernel'' translation that relates the ``kernel'' ANF with CPS. 
\begin{figure}[t]
			\centering
		$$\xymatrix{
			& \ar@{--}[dddd]  &\\
			\boxed{\lbc} \ar@<0ex>[d]_{\text{admin}} \ar@/^1pc/[drr]^{\text{\qquad CPS-translation}}\\
			\boxed{\ANF}  \ar@<1ex>[rr]^{\text{kernel of CPS-translation}}
			&  & \boxed{\RCPS} \ar@<1ex>[ll]^{\text{inverse CPS-translation}} \\
			\small\text{Direct Style} \!\!\!\!\!\!\!\! & & \!\!\!\!\!\!\!\!\small\text{Continuation-Passing Style}\\ 
			&  &}$$
	\caption{The essence of compiling with continuations}
	\label{fig:essence}
\end{figure}
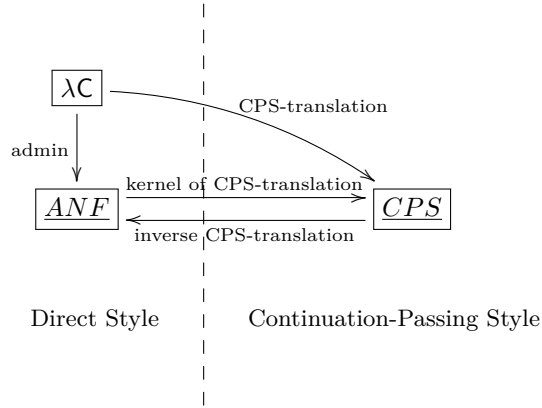


It is also well-known that the CPS-translation can be decomposed in several ways: indeed in the reference \cite{SabryWadler97} alone we may find two of them, one through the monadic meta-language \cite{MoggiIC91}, the other through the linear $\lb$-calculus \cite{MaraistOderskyTurnerWadlerTCS99}. Here we will propose another intermediate language, the sequent calculus $\LJQ$ \cite{DyckhoffLengrandCiE2006,DyckhoffLengrandJLC07}. The calculus $\LJQ$ has a long history and several applications in proof theory \cite{DyckhoffLengrandCiE2006} and can be turned into a typed call-by-value $\lb$-calculus in equational correspondence with $\lbc$ \cite{DyckhoffLengrandJLC07}. Here we want to show it has a privileged role as a tool to analyze the CPS-translation. 

Languages of proof terms for the sequent calculus handle contexts (\emph i.e. $\lb$-terms with a hole) formally \cite{HerbelinCSL94,CurienHerbelinICFP00,jesTOCS09,jesAPAL2013}. This seems most convenient, since a continuation may be seen as a certain kind of context, and suggests that we can write an alternative translation into the sequent calculus, as if we were CPS-translating, but without the need to pass around a reification of the current continuation as a $\lb$-abstraction, nor the concomitant need to translate types by the insertion of double negations, to make room for a type $\ngv{A}$ of values, a type $\neg\ngv{A}$ of continuations and a type $\neg\neg\ngv{A}$ of programs, out of a source type $A$.

We develop this in detail, which requires: to rework entirely the term calculus for $\LJQ$ and obtain a system, named $\lbq$, more manageable for our purposes; and to identify the kernel and the sub-kernel of $\lbq$, the latter being the target system, named $\VFS$ after \emph{value-filling style}, of the new translation. In the end, we are rewarded with an isomorphism between $\VFS$ and the target of the CPS-translation, 
which, when composed with the alternative translation, reconstructs the CPS-translation. The isomorphism is a negative translation, reduced to the role of optional and late stage of translation.

Going back to direct style, the ``essence'' of the VFS-translation is that it reveals a new sublanguage of ANF, the \emph{value-enclosed style} (VES), next to another sublanguage of ANF, the \emph{continuation-enclosing style} (CES): such alternative between VES and CES is due to a dilemma in the syntax of $\lbc$, concerning how to \emph{expand} the application constructor. Hence, these two sub-kernels of $\lbc$ are under a layer of expansion -- and the same was already true for the passage from the kernel to the sub-kernel of $\lbq$. 

While VES corresponds to the sub-kernel VFS of $\lbq$, CES corresponds to a fragment of $\lbjv$ \cite{jesCSL20}, a call-by-value $\lb$-calculus with generalized applications; the fragment is that of \emph{commutative normal forms} (CNF), that is, normal forms w.~r.~t.~the commutative conversions, naturally arising when application is generalized, which reduce both the head term and the argument in an application to the form of values. So the alternative between VES and CES is also a reflection, in the source language, of the alternative between two proof systems for call-by-value: the sequent calculus $LJQ$ and the natural deduction system behind $\lbjv$.

A summary is contained in Fig.~\ref{fig:logical-essence}: it shows a proof-theoretical background hidden in Fig.~\ref{fig:essence}, which this paper wants to reveal. In the process, we want to confirm proof theory as a foundation for intermediate representations useful in the compilation of functional languages.

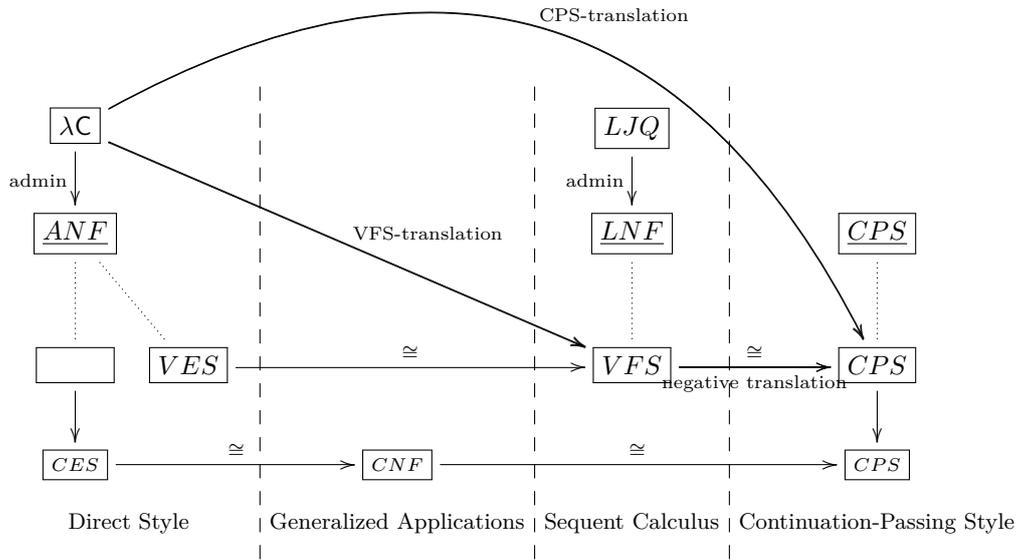
\begin{figure*}[t]
		\centering
		$$\xymatrix@=6pt{ &  &  \ar@{--}[dddddddddd] &         & \ar@{--}[dddddddddd]  & & \ar@{--}[dddddddddd]  & \\
						\boxed{\lbc}  \ar@<0ex>[dd]_{\text{admin}}  \ar@/^7.2pc/@{->}[rrrrrrrddddd]^{\text{\qquad\,\, \, CPS-translation}} \ar@/^7.2pc/@{->}@<0.1mm>[rrrrrrrddddd] \ar@{->}[dddddrrrrr]^{\qquad \text{VFS-translation}} \ar@{->}@<0.1mm>[dddddrrrrr] &  &  &  & & \boxed{\LJQ}  \ar@<0ex>[dd]_{\text{admin}}  &\\
						& & & & & &\\
						\boxed{\ANF} \ar@{..}[ddd] \ar@{..}[dddr] & & & & & \boxed{\LNF} \ar@{..}[ddd] & & \boxed{\RCPS} \ar@{..}[ddd] \\
						& & & & & &\\
						& & & & & &\\
			   		   	\boxed{\textcolor{white}{VES}} \ar@<0ex>[dd] & \boxed{\VES} \ar@{->}[rrrr]^{\cong} & &  & & \boxed{\VFS} \ar@{->}[rr]_{\text{negative translation}}^{\cong}  \ar@{->}@<0.1mm>[rr]& & \boxed{\CPS} \ar@<0ex>[dd] \\
			   		   	& & & & & & \\
						\boxed{\ces} \ar@{->}[rrr]^{\cong} & && \boxed{\cnf}  \ar@{->}[rrrr]^{\cong}   & & & & \boxed{\cps} \\
						\!\!\!\!\! \footnotesize\text{Direct}  \!\!\!\!\!\!\!\!\!\!\!\!\!\!\!\! & \!\!\!\!\!\!\!\!\!\!\!\!\!\!\!\! \footnotesize\text{Style}  \!\!\!\!\! & & \!\!\!\!\! \footnotesize\text{Generalized Applications}\!\!\!\!\! & &\!\!\!\!\!  \footnotesize\text{Sequent Calculus} \!\!\!\!\! & & \!\!\!\!\! \footnotesize\text{Continuation-Passing Style} \!\!\!\!\!  \\
						& & & & & & }$$
	\caption{The logical essence of compiling with continuations}
	\label{fig:logical-essence}
\end{figure*}

\textbf{Plan of the paper.} Section \ref{sec:back} recalls $\lbc$ and the CPS-translation. Section \ref{sec:LJQ} contains our reworking of $LJQ$. Section \ref{sec:VFS} introduces the alternative translation into $LJQ$ and the decomposition of the CPS-translation. Section \ref{sec:bds} goes back to direct style and studies the sub-kernels of $\lbc$. Section \ref{sec:conc} summarizes our contribution and 
discusses related and future work. All proofs \whereproofs{are moved to Appendix \ref{sec:proofs}}{can be found in the long version of this paper \cite{}}.

%% file: background.tex
\section{Background}\label{sec:back}

\textbf{Preliminaries.} Simple types (=formulas) are given by $A,B,C::= a | A\imp B$. In typing systems, a context $\Gamma$ will always be a \emph{consistent} set of declarations $x:A$; consistency here means that no variable can be declared with two different types in $\Gamma$. 


We recall the concepts of equational correspondence, pre-Galois connection and reflection \cite{DyckhoffLengrandJLC07,SabryFelleisen93,SabryWadler97} characterizing different forms of relationship between two calculi.

\begin{defn}
	Let $(\Lambda_1,\red_1)$ and $(\Lambda_2,\red_2)$ be two calculi and, for each $i=1,2$, let $\redd_i$ (resp.~$\conv_i$) be the reflexive-transitive (resp.~reflexive-transitive-symmetric) closure of $\red_i$. Consider the mappings $f: \Lambda_1 \to \Lambda_2$ and $g: \Lambda_2 \to \Lambda_1$.
	\begin{itemize}
		\item $f$ and $g$ form an \textbf{equational correspondence} between $\Lambda_1$ and $\Lambda_2$ if the following conditions hold: (1) If $ M \red_1 N$ then $f(M) \conv_2 f(N)$; (2) If $ M \red_2 N$ then $g(M) \conv_1 g(N)$; (3) $ M \conv_1 g(f(M)) $; (4) $ f(g(M)) \conv_2 M$.
		\item $f$ and $g$ form a \textbf{pre-Galois connection} from $\Lambda_1$ to $\Lambda_2$ if the following conditions hold: (1) If $ M \red_1 N $ then $ f(M) \redd_2 f(N)$; (2) If $ M \red_2 N$ then  $g(M) \redd_1 g(N)$; (3) $ M \redd_1 g(f(M)) $.		
		\item $f$ and $g$ form a \textbf{reflection} in $\Lambda_1$ of $\Lambda_2$ if the following conditions hold: (1) If $ M \red_1 N $ then $ f(M) \redd_2 f(N)$; (2) If $ M \red_2 N $ then $ g(M) \redd_1 g(N)$; (3) $ M \redd_1 g(f(M)) $; (4) $ f(g(M)) \eq M$.
		
	\end{itemize}
	
\end{defn}
Note that if $f$ and $g$ form a pre-Galois connection from $\Lambda_1$ to $\Lambda_2$ and $\red_2$ is confluent, then $\red_1$ is also confluent.  Besides, it is also important to observe that if $f$ and $g$ form a reflection from $\Lambda_1$ to $\Lambda_2$, then $g$ and $f$ form a pre-Galois connection from $\Lambda_2$ to $\Lambda_1$.

\textbf{Computational lambda-calculus.} The computational $\lb$-calculus \cite{MoggiLFCS88} is defined in Table \ref{tab:lC}. In addition to ordinary $\lb$-terms, one also has \emph{let-expressions} $\lt{x}{M}{N}$: these are explicit substitutions which trigger only after the actual parameter $M$ is reduced to a value (that is, a variable or $\lb$-abstraction). So, in addition to the rule $\ltv$ that triggers substitution, there are reduction rules --  $\ltmn$, $\ltvn$ and $\assoc$ -- dedicated to that preliminary reduction of actual parameters in let-expressions. 

For the reduction of $\beta$-redexes, we adopt the rule $B$ from \cite{DyckhoffLengrandJLC07}, which triggers even if the argument $N$ is not a value, and just generates a let-expression. Most presentations of $\lbc$ \cite{MoggiLFCS88,SabryWadler97} have rule $\betav$ instead, which reads $(\lb x.M)V\to\sub{V}{x}{M}$. The two versions of the system are equivalent. In our presentation, the effect of $\betav$ is achieved with $B$ followed by $\ltv$. Conversely, when $N$ is not a value, we can perform the reduction 
$$(\lb x.M)N\to\lt{y}{N}{(\lb x.M)y}\to\lt{y}{N}{\sub{y}{x}{M}}=_{\alpha}\lt{x}{N}{M}\enspace.$$ 
The first step is by $\ltvn$, the second by $\betav$. The last term is the contractum of $B$. 

In this paper, we leave the $\eta$-rule for $\lb$-abstraction out of the definition of $\lbc$, and similarly for other systems -- since it plays no rule in what we want to say. But we include the $\eta$-rule for let-expressions, and other incarnations of it in other systems.

In \cite{DyckhoffLengrandJLC07,SabryWadler97} the $\lbc$-calculus is studied in its untyped version. Here we will also consider its simply-typed version, which handles sequents $\seqC{\Gamma}{M}{A}$, where $\Gamma$ is a set of declarations $x:A$. The typing rules are obvious, Table \ref{tab:lC} only contains the rule for typing let-expressions.

\begin{table*}
		\centering
		
		$$\begin{array}{r r c l}
			\text{(terms)} & M, N, P, Q &::= & V  \, | \,  MN \, | \, \lt{x}{M}{N} \\
			\text{(values)} & V,W &::=& x \, | \, \lb x.M 
		\end{array}$$
%
		$$ \begin{array}{rrcll}
			(B)	& (\lb x.M)N& \red & \lt{x}{N}{M}  \\
			(\ltv) & \lt{x}{V}{M} & \red & \sub{V}{x}{M}\\
			(\etalt) & \lt{x}{M}{x} & \red & M \\
			(\assoc) & \lt{y}{(\lt{x}{M}{N})}{P} & \red & \lt{x}{M}{\lt{y}{N}{P}} \\		
			(\ltmn) &  MN & \red &  \lt{x}{M}{xN} & (a)\\
			(\ltvn) &  VN & \red &  \lt{x}{N}{Vx} & (b)
		\end{array}$$
%
%
	   $$
	   \begin{array}{l}
	   	\infer{\seqC{\Gamma}{\lt{x}{M}{N}}{B}}{\seqC{\Gamma}{M}{A}&\seqC{\Gamma,x:A}{N}{B}}
	   \end{array}
	  $$	
	\caption{The computational $\lb$-calculus, here also named $\lbc$-calculus. Provisos: $(a)$ $M$ is not a value. $(b)$ $N$ is not a value. Typing rules for $x$, $\lb x.M$ and $MN$ as usual.}
	\label{tab:lC}
\end{table*}

The \emph{kernel} of the computational $\lb$-calculus \cite{SabryWadler97} is defined in Table \ref{tab:kernel-lC}. It is named here $\ANF$, after ``administrative normal form'', because its terms are the normal forms w.~r.~t.~the administrative rules of $\lbc$: $\ltmn$, $\ltvn$ and $\assoc$ \cite{SabryWadler97}.

\begin{table*}
		\centering
		
		$$\begin{array}{r r c l}
			\text{(terms)} & M, N, P, Q &::= & V  \, | \,  VW \, | \, \lt{x}{V}{M} \, | \, \lt{x}{VW}{M} \\ 
			\text{(values)} & V,W &::=& x \, | \, \lb x.M 
		\end{array}$$
%
		$$ \begin{array}{rrcll}
			(\Bv)	& \lt{y}{(\lb x.M)V}{P} & \red & \lt{x}{V}{\ltc{y}{M}{P}} & \\
			(\Bv')	& (\lb x.M)V & \red & \lt{x}{V}{M} & \\
			(\ltv) & \lt{x}{V}{M} & \red & \sub{V}{x}{M}\\
			(\etalt) & \lt{x}{VW}{x} & \red & VW 
		\end{array}$$
%
%
	\caption{The kernel of the computational $\lb$-calculus, here named $\ANF$.}
	\label{tab:kernel-lC}
\end{table*}

In the kernel, only a specific form of applications and two forms of let-expressions are primitive. The general form of a let-expression, written $\ltc{y}{M}{P}$, is a derived form defined by recursion on $M$ as follows: 
$$
\begin{array}{rcl}
	\ltc{y}{V}{P} &=& \lt{y}{V}{P}\\
	\ltc{y}{VW}{P} &=& \lt{y}{VW}{P}\\
	\ltc{y}{(\lt{x}{V}{M})}{P} &=& \lt{x}{V}{\ltc{y}{M}{P}}\\
	\ltc{y}{(\lt{x}{VW}{M})}{P} &=& \lt{x}{VW}{\ltc{y}{M}{P}}
\end{array}
$$

Obviously, given $M$ and $P$ in the kernel, $\lt{y}{M}{P}\redd_{\assoc}\ltc{y}{M}{P}$ in $\lbc$. Hence, a $\Bv$-step in the kernel can be simulated in $\lbc$ as a $B$-step followed by a series of $\assoc$-steps. On the other hand $\Bv'$ is a restriction of rule $B$ to the sub-syntax, and the same is true of the remaining rules of the kernel. 

Notice that in the form $\lt{x}{VW}{M}$ the immediate sub-expressions are $V$, $W$ and $M$ -- but not $VW$. For this reason, there is no overlap between the redexes  of rules $\Bv$ and $\Bv'$, nor between the redexes of rules $\Bv'$ and $\etalt$.

Our presentation of the kernel is very close to the original one in \cite{SabryWadler97}, as detailed in Appendix \ref{sec:kernel-lbc}.

\textbf{CPS-translation.} We present in this subsection the call-by-value CPS-translation of $\lbc$. It is a ``refined '' translation \cite{DyckhoffLengrandJLC07}, in the sense that it reduces ``administrative redexes'' at translation time, as already done in \cite{PlotkinTCS75}.

The target of the translation is the system $\RCPS$, presented in Table \ref{tab:RCPS}. This target is a subsystem of the $\lb$-calculus (or of Plotkin's call-by-value $\lb_v$-calculus -- the ``indifference property'' \cite{PlotkinTCS75}), whose expressions are the union of four different classes of $\lb$-terms (commands, continuations, values and terms), and whose reduction rules are either particular cases of rules $\beta$ and $\eta$ (the cases of $\sigmav$ or $\etak$, respectively), or are derivable as two $\beta$-steps (the case of $\betav$). Each command or continuation has a unique free occurrence of $k$, which is a fixed (in the calculus) \emph{continuation variable}. A term is obtained by abstracting this variable over a command. A command is always composed of a continuation $K$, to which a value may be passed (the form $KV$), or which is going to instantiate $k$ in the command resulting from an application $VW$ (the form $VWK$).

There is a simply-typed version of this target, \emph{not} found in \cite{PlotkinTCS75,DyckhoffLengrandJLC07,SabryWadler97}, defined as follows. Simple types are augmented with a new type $\perp$, and we adopt the usual abbreviation $\neg A:=A\imp\perp$. Then, as defined in Table \ref{tab:RCPS}, one has: two subclasses of such types, one ranged by $\ac$, $\ac'$ and the other ranged over by $\bc$, $\bc'$; four kinds of sequents, one per each syntactic class; and one typing rule for each syntactic constructor.

\begin{table*}
		$$
		\begin{array}{rrcl}
		\text{(Commands)} & M,N & ::= & KV \,|\, VWK\\
		\text{(Continuations)} & K & ::= & \lb x.M\,|\, k \\
		\text{(Values)} & V,W & ::= & \lb x.P \,|\, x\\
		\text{(Terms)} & P & ::= & \lb k.M
		\end{array}
		$$
		$$
		\begin{array}{rrcll}
			(\sigmav) & (\lb x.M)V & \to & \sub{V}{x}{M}\\
			(\betav)  & (\lb xk.M)WK & \to & \sub{K}{k}{\sub{W}{x}{M}}\\
			(\etak)  & \lb x.Kx & \to & K & \text{if $x\notin FV(K)$}
		\end{array}
		$$
		$$\text{Types:} \qquad \ac ::= a \,|\, \ac\imp\bc \qquad \bc ::= \neg\neg \ac$$
		$$\text{Contexts $\Gamma$: sets of declarations $(x:\ac)$} $$
		$$\text{Sequents:} \quad \seqCPS{k:\neg\ac,\Gamma}{M}{\perp} \quad \seqCPS{k:\neg\ac,\Gamma}{K}{\neg\ac'} \quad \seqCPS{\Gamma}{V}{\ac} \quad \seqCPS{\Gamma}{P}{\bc}
		$$

		$$
		\begin{array}{c}
		\infer{\seqCPS{k:\neg\ac,\Gamma}{KV}{\perp}}{\seqCPS{k:\neg\ac,\Gamma}{K}{\neg\ac'}&\seqCPS{\Gamma}{V}{\ac'}}
	\\ \\
		\infer{\seqCPS{k:\neg\ac'',\Gamma}{VWK}{\perp}}{\seqCPS{\Gamma}{V}{\ac\imp\neg\neg\ac'}&\seqCPS{\Gamma}{W}{\ac}&\seqCPS{k:\neg\ac'',\Gamma}{K}{\neg\ac'}}\\ \\
		\infer{\seqCPS{k:\neg\ac,\Gamma}{\lb x.M}{\neg\ac'}}{\seqCPS{k:\neg\ac,\Gamma,x:\ac'}{M}{\perp}}
		\qquad
		\infer{\seqCPS{k:\neg\ac,\Gamma}{k}{\neg\ac}}{}
		\\ \\
		\infer{\seqCPS{\Gamma}{\lb x.P}{\ac\imp\bc}}{\seqCPS{\Gamma,x:\ac}{P}{\bc}}
		\qquad		
		\infer{\seqCPS{\Gamma,x:\ac}{x}{\ac}}{}
		\\ \\
		\infer{\seqCPS{\Gamma}{\lb k.M}{\neg\neg\ac}}{\seqCPS{k:\neg\ac,\Gamma}{M}{\perp}}
		\end{array}
		$$
	\caption{The system $\RCPS$}
	\label{tab:RCPS}
\end{table*}
The CPS-translation is defined in Table \ref{tab:cps-translation}. It comprises: For each $V\in\lbc$, a value $\cpsv{V}$; for each term $M\in\lbc$ and continuation $K\in\RCPS$, a command $\cpsc{M}{K}$; for each term $M\in\lbc$, a command $\cpsk{M}$ and a term $\cpst{M}$.

In the typed setting, each simple type $A$ of $\lbc$ determines an $\ac$-type $\cpsv{A}$ and a $\bc$-type $\cpst{A}$, as in Table \ref{tab:cps-translation}. The translation preserves typing, according to the admissible typing rules displayed in the last row of the same table.

\begin{table*}
		\centering
		$$
		\begin{array}{rclcrcll}
			\cpsv x & = & x &\qquad&  \cpsc V K& = & K\cpsv V\\
			\cpsv{(\lb x.M)} & = & \lb x.\cpst M && \cpsc{PQ}{K}& = & \cpsc{P}{\lb m.\cpsc{mQ}{K}} & (a)\\
			\cpst M & = & \lb k.\cpsk M && \cpsc{VQ}{K}& = & \cpsc{Q}{\lb n.\cpsc{Vn}{K}} & (b)\\
			\cpsk M & = & \cpsc M k & & \cpsc{VW}{K} & = & \cpsv{V}\cpsv{W} K\\
			& &  && \cpsc{\lt{y}{M}{P}}{K} & = & \cpsc{M}{\lb y.\cpsc{P}{K}}
		\end{array}
		$$
%
%
		$$
		\cpst{A}=\neg\neg\cpsv{A} \qquad \cpsv{a}=a \qquad \cpsv{(A\imp B)}=\cpsv{A}\imp\cpst{B}
		$$
		$$
		\begin{array}{cc}
			\infer{\seqCPS{\cpsv{\Gamma}}{\cpsv{V}}{A}}{\seqC{\Gamma}{V}{A}}
			&
			\infer{\seqCPS{k:\neg\cpsv{B},\cpsv{\Gamma}}{\cpsc{M}{K}}{\perp}}{\seqC{\Gamma}{M}{A}&\seqCPS{k:\neg\cpsv{B},\Gamma}{K}{\neg\cpsv{A}}}
			\\ \\
			\infer{\seqCPS{k:\neg\cpsv{A},\cpsv{\Gamma}}{\cpsk{M}}{\perp}}{\seqC{\Gamma}{M}{A}}
			&
			\infer{\seqCPS{\cpsv{\Gamma}}{\cpst{M}}{\cpst{A}}}{\seqC{\Gamma}{M}{A}}
		\end{array}
		$$
	\caption{The CPS-translation, from $\lbc$ to $\RCPS$, with admissible typing rules. Provisos: $(a)$ $P$ is not a value. $(b)$ $Q$ is not a value.}
	\label{tab:cps-translation}
\end{table*}

%% file: LJQ.tex
\section{Sequent calculus $\LJQ$ and its simplification $\lbq$}\label{sec:LJQ}

In this section we start by recapitulating the term calculus for $\LJQ$ designed by Dyckhoff-Lengrand \cite{DyckhoffLengrandJLC07}. Next we do some preliminary work, by proposing a simplified variant, named $\lbq$, more appropriate for our purposes in this paper. Finally, we also single out the \emph{kernel} of $\lbq$, which is the sub-calculus of ``administrative'' normal forms. This further simplification will be necessary for the later analysis of CPS.

\textbf{The original term calculus.} An abridged presentation of the original term calculus for $\LJQ$ by Dyckhoff-Lengrand is found in Table \ref{tab:abridged-original-lQ} \footnote{See Appendix \ref{sec:original} for the full system.}. The separation between terms and values corresponds to the separation between the two kinds of sequents handled by $\LJQ$: the ordinary sequents $\seq{\Gamma}{M}{A}$ and the focused sequents $\seqv{\Gamma}{V}{A}$. There are three forms of cut and the reduction rules correspond to cut-elimination rules. We may think of the forms $\CVW{V}{x}{W}$ and $\CVN{V}{x}{N}$ as explicit substitutions: in this abridged presentation we omitted the rules for their stepwise execution.

\begin{table*}	
		
		$$\begin{array}{r r c l}
		 	 \text{(terms)} & M, N &::= &\rt{V}  \, | \,  \li{x}{V}{y}{N} \, | \, \CVN{V}{x}{N} \, | \,  \CMN{M}{x}{N} \\
			 \text{(values)} & V,W &::=& x \, | \, \lb x.M \, | \, \CVW{V}{x}{W}
		\end{array}$$
	
		$$ \begin{array}{rrcll}
		(1)	& \CMN{\rt{(\lb x.M)}}{y}{\li{y}{V}{z}{N}} &\to&  \CMN{\CMN{\rt{V}}{x}{M}}{z}{N} & (a) \\
		(2)	& \CMN{\rt{x}}{y}{N} &\to & \sub{x}{y}{N}  \\
		(3)	& \CMN{M}{x}{\rt{x}} &\to & M   \\
		(4)	& \CMN{\li{z}{V}{y}{P}}{x}{N} &\to & \li{z}{V}{y}{\CMN{P}{x}{N}} \\
		(5)	& \CMN{\CMN{\rt{W}}{y}{\li{y}{V}{z}{P}}}{x}{N} &\to & \CMN{\rt{W}}{y}{\li{y}{V}{z}{\CMN{P}{x}{N}}} & (b) \\
		(6)	& \CMN{\CMN{M}{y}{P}}{x}{N} &\to &\CMN{M}{y}{\CMN{P}{x}{N}} & (c)\\
		(7)	& \CMN{\rt({\lb x.M})}{y}{N} &\to & \CVN{\lb x.M}{y}{N} & (d)\\
		\end{array}$$
%
%
		$$\begin{array}{c c}
			\infer[Ax]{\seqv{\Gamma, x:A}{x}{A}}{} & \infer[Der]{\seq{\Gamma}{\rt{V}}{A}}{\seqv{\Gamma}{V}{A}}\\
			\\
			\infer[R\!\imp]{\seqv{\Gamma}{\lb x.M}{A \imp B}}{\seq{\Gamma, x:A}{M}{B}} &  \infer[Cut_3]{\seq{\Gamma}{\CMN{M}{x}{N}}{B}}{\seq{\Gamma}{M}{A} & \seq{\Gamma, x:A}{N}{B}}\\
			\\
			\multicolumn{2}{c}{\infer[L\!\imp]{\seq{\Gamma, x: A \imp B}{\li{x}{V}{y}{N}}{C}}{\seqv{\Gamma, x:A \imp B}{V}{A} & \seq{\Gamma, x:A \imp B, y:B}{N}{C}}} 
		\end{array}$$
		
	\caption{The original calculus by Dyckhoff-Lengrand, here named $\lbqo$-calculus (abridged). Provisos: $(a)$ $y \notin FV(V) \cup FV(N)$. $(b)$ $y \notin FV(V) \cup FV(P))$. $(c)$ If rule (5) does not apply. $(d)$ If rule (1) does not apply.}
	\label{tab:abridged-original-lQ}
\end{table*}

We now introduce a slight modification of $\lbqo$, named $\lbqmo$, determined by two changes in the reduction rules: in rule (6) we omit the proviso; and rule (5) is dropped. A former redex of (5) is reduced by (6) -- now possible because there is no proviso -- followed by (4), achieving the same effect as previous rule (5).

In fact, very soon we will define a \emph{big} modification and simplification of the original $\lbqo$, which is more appropriate to our goals here. But we need to justify that big modification, by a comparison with the original system. For the purpose of this comparison, we will use, not $\lbqo$, but $\lbqmo$ instead. So, the first thing we do is to check that $\lbqmo$ has the same properties as the original.

The maps between $\lbc$ and $\lbqo$ defined by Dyckhoff-Lengrand can be seen as maps to and from $\lbqmo$ instead. Next, it is easy to see that such maps still establish an equational correspondence, now  between $\lbc$ and $\lbqmo$. It turns out that the correspondence is also a pre-Galois connection from $\lbqmo$ to $\lbc$. Because of this, $\lbqmo$ inherits confluence of $\lbc$, as $\lbqo$ did. 


%
%
%

\textbf{A simplified calculus.} 
We now define the announced simplified calculus, named $\lbq$. It is presented in Table \ref{tab:simplified-lQ}. The idea is to drop the cut forms $\CVW{V}{x}{W}$ and $\CVN{V}{x}{N}$, which correspond to explicit substitutions. Since only one form of cut remain, $\CMN{M}{x}{N}$, we write it as $\Cut MxN$. The typing rules of the surviving constructors remain the same. The omitted reduction rules for the stepwise execution of substitution are now dropped, since they concerned the omitted forms of cut. Rules (1) and (3) are renamed as $\Bv$ and $\etact$, respectively. Rules (4) and (6) are renamed $\pi_1$ and $\pi_2$, respectively, and we let $\pi:=\pi_1\cup\pi_2$. Rules (2) and (7) are combined into a single rule named $\sigmav$.

\begin{table*}
		\centering
		
		$$\begin{array}{rrcl}
			\text{(terms)}& M, N &::= &\rt{V}  \, | \,  \li{x}{V}{y}{N}  \, | \,  \Cut{M}{x}{N} \\
			\text{(values)}& V,W &::= &x \, | \, \lb x.M
		\end{array}$$		
%
		$$\begin{array}{rrcll}
			(\Bv) & \Cut{\rt{(\lb x.M)}}{y}{\li{y}{V}{z}{N}} & \to & \Cut{\Cut{\rt{V}}{x}{M}}{z}{N} & \text{if $y \notin FV(V) \cup FV(N)$} \\
			(\sigmav) & \Cut{\rt{V}}{y}{N} & \to & \sub{V}{y}{N} & \text{if $\Bv$ does not apply}\\
			(\etact) & \Cut{M}{x}{\rt{x}} & \to & M   \\
			(\pi_1) & \Cut{\li{z}{V}{y}{P}}{x}{N} & \to & \li{z}{V}{y}{\Cut{P}{x}{N}} \\
		    (\pi_2) & \Cut{\Cut{M}{y}{P}}{x}{N}  & \to & \Cut{M}{y}{\Cut{P}{x}{N}}
		\end{array}$$
	\caption{The simpified $\lbqo$-calculus, named $\lbq$-calculus}
	\label{tab:simplified-lQ}
\end{table*}

The design of rule $\sigmav$ is interesting. Rule (2) fired a variable substitution operation $\sub{x}{y}{-}$, already present in the original calculus. The contractum of rule (7), being an explicit substitution, has to be replaced by the call to an appropriate, \emph{implicit}, substitution operator $\sub{\lb x.M}{y}{-}$, whose stepwise execution should be coherent with the omitted reduction rules for $\CVW{V}{x}{W}$ and $\CVN{V}{x}{N}$. Hopefully, the sought operation and the already present variable substitution operation are subsumed by a value substitution operation $\sub{V}{y}{-}$.

The critical clause is the definition of $\sub{V}{y}{(\li{y}{W}{z}{P})}$. We adopt
$\sub{V}{y}{(\li{y}{W}{z}{P})}=\Cut{\rt{V}}{y}{\li{y}{\sub{V}{y}{W}}{z}{\sub{V}{y}{P}}}$ 
in the case $V=\lb x.M$, but not in the case of $V=x$, because $\sigmav$ would immediately generate a cycle in the case $y\notin FV(V)\cup FV(N)$. We adopt instead
$\sub{x}{y}{(\li{y}{W}{z}{P})}=\li{x}{\sub{x}{y}{W}}{z}{\sub{x}{y}{P}}$ 
which moreover is what the original calculus dictates. Notice that another cycle would arise, if a $\Bv$-redex was contracted by $\sigmav$. But this is blocked by the proviso of the latter rule.

There is a map $\smp{(\_)}:\lbqmo\to\lbq$, based on the idea of translating the omitted cuts by calls to substitution: $\CVW{V}{x}{W}$ is mapped to $\sub{V}{x}{W}$ and $\CVN{V}{y}{N}$ is mapped to $\sub{V}{x}{N}$. This map, together with the inclusion $\lbq\subset\lbqmo$ (seeing $\Cut{M}{x}{N}$ as $\CMN{M}{x}{N}$) gives a reflection of $\lbq$ in $\lbqmo$. This reflection allows to conclude easily that reduction in $\lbqmo$ is conservative over reduction in $\lbq$. Moreover, this reflection can be composed with the equational correspondence between $\lbc$ and $\lbqmo$ to produce an equational correspondence between $\lbc$ and $\lbq$. Finally, this reflection is also a pre-Galois connection from $\lbq$ to $\lbqmo$. Thus, confluence of $\lbq$ can be pulled back from the confluence of $\lbqmo$.
To sum up, we obtained a more manageable calculus, conservatively extended by the original one, which, as the latter, is confluent and is in equational correspondence with $\lbc$.	

\textbf{The kernel of the simplified calculus.} For a moment, we do an analogy between $\lbc$ and $\lbq$. As was recalled in Section \ref{sec:back}, the former system admits a \emph{kernel}, a subsystem of ``administrative'' normal forms, which are the normal forms with respect to a subset of the set of reduction rules \cite{SabryWadler97}. For $\lbq$, the ``administrative'' normal forms are very easy to characterize: in a cut $\Cut{M}{x}{N}$, $M$ has to be of the form $\rt{V}$. Logically, this means that the left premiss of the cut comes from a sequent $\seqv{\Gamma}{V}{A}$; given that such sequents are obtained either with $Ax$ or $R\!\imp$, the cut formula $A$ in that premiss is not a passive formula of the previous inference; hence the cut is fully permuted to the left -- so we call such forms \emph{left normal forms}. The reduction rules of $\lbq$ which perform left permutation are rules $\pi_1$ and $\pi_2$ (even though textually the outer cut in the redex of those rules seems to move to the right after the reduction), so these rules are declared ``administrative''.

The kernel of $\lbq$ is named $\LNF$. The specific form of cut allowed, namely $\Cut{\rt{V}}{x}{N}$, is written $\Cutv{V}{x}{N}$. No other change is made to the grammar of terms. Given $M,N\in\LNF$, the general form of cut becomes in $\LNF$ a derived constructor written $\Cutvc{M}{z}{N}$ and defined by recursion on $M$ as follows: 
$$
\begin{array}{rcl}
\Cutvc{\rt{V}}{z}{N}&=&\Cutv{V}{z}{N}\\
\Cutvc{\li{x}{V}{y}{M}}{z}{N}&=&\li{x}{V}{y}{\Cutvc{M}{z}{N}}\\
\Cutvc{\Cutv{V}{y}{M}}{z}{N}&=&\Cutv{V}{y}{\Cutvc{M}{z}{N}}
\end{array}
$$
As to reduction rules, rule $\Bv$ in $\LNF$ reads 
$$\Cutv{\lb x.M}{y}{\li{y}{V}{z}{N}} \to \Cutv{V}{x}{\Cutvc{M}{z}{N}}\enspace.$$
Notice that the contractum is the same as $\Cutvc{\Cutvc{\rt{V}}{x}{M}}{z}{N}$. The proviso remains the same: $y \notin FV(V) \cup FV(N)$.
As to the other reduction rules: there is no change to rule $\sigmav$; the specific form of rule $\etact$ that survives becomes a particular case of $\sigmav$, hence is omitted; and the system has no $\pi$-rules. 

%
%


There is a map $\knl{(\_)}:\lbq\to\LNF$ based on the idea of replacing $\Cut{M}{x}{N}$ by $\Cutvc{M}{x}{N}$. This map, together with the inclusion $\LNF\subset\lbq$ (seeing $\Cutv{V}{x}{N}$ as $\Cut{\rt V}{x}{N}$), gives a reflection in $\lbq$ of $\LNF$. Quite obviously, $M\redd_{\pi}\knl{M}$; in fact $\knl{M}$ is a $\pi$-normal form, as are all the expressions of $\LNF$.

$\LNF$ is a stepping stone in the way to the definition, in the next section, of the value-filling style fragment, which will be a central player in this paper.
%
%
%

%% file: VFS.tex
\section{The value-filling style}\label{sec:VFS}

In this section we define the target language $\VFS$ (a fragment of $\LNF$) of a new compilation of $\lbc$, the value-filling style translation. Next we slightly modify the target $\RCPS$, and introduce the negative translation, mapping $\VFS$ to the modified $\RCPS$. Then we show that the CPS-translation is decomposed in terms of the alternative compilation and the negative translation; and that the negative translation is in fact an isomorphism.

\textbf{The sub-kernel of $\LJQ$.} We now define the \emph{sub-kernel} of $\lbq$, a language named $\VFS$ that will serve as a target language for compilation alternative to $\CPS$. Despite the simplicity of $\lbq$, there is still room for a huge simplification: to forbid the left-introduction constructor $\li{y}{W}{x}{M}$ to stand as a term on its own. However, we regret that, by that omission, that term cannot be used in a very particular situation: as the term $N$ in $\Cutv{V}{y}{N}$, when $y\notin FV(W)\cup FV(M)$. So, we keep that particular combination of cut and left-introduction as a separate form of cut. The result is presented in Table \ref{tab:sub-kernel-lQ}.

\begin{table*}
		\centering
		
		$$\begin{array}{rrcl}
			\text{(terms)} & M, N &::= & \rt{V} \, | \,  \Cutvfs{V}{c} \\
			\text{(values)}& V,W &::=& x \, | \, \lb x.M\\
			\text{(formal contexts)} & c &::= &x.M  \, | \, \garg{W}{x}{M} 
		\end{array}$$
%
		$$\begin{array}{rrcll}
			(\Bv) & \Cutvfs{\lb x.M}{\garg{V}{y}{N}} &  \to & \Cutv{V}{x}{\Cutvc{M}{y}{N}}\\
			(\sigmav) & \Cutv{V}{y}{N} & \to & \sub{V}{y}{N} & 
		\end{array}$$
%
	$$
	\begin{array}{ccc}
	\infer{\seq{\Gamma}{\Cutvfs{V}{c}}{B}}{\seqv{\Gamma}{V}{A}&\seqc{\Gamma}{A}{c}{B}}	
    &
    \infer{\seqc{\Gamma}{A}{x.M}{B}}{\seq{\Gamma,x:A}{M}{B}}
    &
    \infer{\seqc{\Gamma}{A\imp B}{\garg{W}{x}{M}}{C}}{\seqv{\Gamma}{W}{A}&\seq{\Gamma,x:B}{M}{C}}
    \end{array}
	$$		
	\caption{The sub-kernel of the $\lbq$, 
		named $\VFS$. Typing rules for $\rt{V}$, $x$ and $\lb x.M$ as before.}
	\label{tab:sub-kernel-lQ}
\end{table*}
In fact, we introduce a third syntactic class, that of \emph{formal contexts} -- this terminology will be justified later. Think of $\garg{W}{x}{M}$ as $y.\li{y}{W}{x}{M}$ with $y\notin FV(W)\cup FV(M)$. The new class allows us to account uniformly for the two possible forms of cut: $\Cutvfs{V}{c}$. The reduction rules of $\VFS$ are those of the kernel $\LNF$, restricted to the sub-kernel: pleasantly, the side conditions have vanished! Moreover, the operation $\sub{V}{y}{N}$ is now plain substitution.

There is, again, an auxiliary operation used in the contractum of $\Bv$. Cut $\Cutvfsc{M}{c'}$ and formal context $(c:c')$ are defined by simultaneous recursion on $M$ and $c$ as follows: 
$$
\begin{array}{rclcrcl}
\Cutvfsc{\rt V}{c'}&=&\Cutvfs{V}{c'}&&((x.M):c')&=&x.\Cutvfsc{M}{c'}\\
\Cutvfsc{\Cutvfs{V}{c}}{c'}&=&\Cutvfs{V}{(c:c')}&&(\garg{W}{x}{M}:c')&=&\garg{W}{x}{\Cutvfsc{M}{c'}}
\end{array}
$$

In the type system, a third form of sequents is added for the typing of formal contexts. We know the formula $A$ in $\seqv{\Gamma}{V}{A}$ is a focus \cite{DyckhoffLengrandJLC07}, but the formula $A$ in $\seqc{\Gamma}{A}{c}{B}$ is not, since it can simply be selected from the context $\Gamma$ in the typing rule for $x.M$.

We already know how to map $\VFS$ back to $\LNF$. How about the inverse direction? How do we compensate the omission of $\li{y}{W}{x}{M}$? The answer is: by the following expansion
\begin{equation}\label{eq:expansion-vfs}
	\li{y}{W}{x}{M}\leftarrow_{\sigmav}\Cutv{y}{z}{\li{z}{W}{x}{M}}=\Cutvfs{y}{\garg{W}{x}{M}}
\end{equation}

\textbf{The VFS-translation.} The system $\VFS$ is the target of a translation of $\lbc$ alternative to the CPS-translation, to be introduced now. The idea is to represent a term of $\lbc$, not as a command of $\CPS$ (in terms of a continuation that is called of passed), but rather as a cut of the sequent calculus $\VFS$, making use of ``formal contexts''. Later, we will give a detailed comparison with the CPS-translation, which will make sense of the terminology ``formal context'' and ``value-filling''; more importantly, the comparison will show that $\VFS$ and the translation into it is a style equivalent to CPS, but much simpler, in particular due to this very objective fact: there is no translation of types involved.

The VFS-translation is given in Table \ref{tab:vfs-translation}. It comprises: For each $V\in\lbc$, a value $\vfsv{V}$ in $\VFS$; for each $M\in\lbc$ and formal context $c\in\VFS$, a cut $\vfsc{M}{c}$ in $\VFS$; for each $M\in\lbc$, a cut $\vfst{M}$ in $\VFS$. Again: there is no translation of types.

\begin{table*}
	$$
	\begin{array}{rclcrcll}
		\vfsv x & = & x &\qquad&  \vfsc V{x.N}& = & \Cutvfs{\vfsv V}{x.N}\\
		\vfsv{(\lb x.M)} & = & \lb x.\vfst M && \vfsc{PQ}{x.N}& = & \vfsc{P}{m.\vfsc{mQ}{x.N}} & (*)\\
		& &  && \vfsc{VQ}{x.N}& = & \vfsc{Q}{n.\vfsc{Vn}{x.N}} & (**)\\
		\vfst M & = & \vfsc M{x.\rt{x}} & & \vfsc{VW}{x.N} & = & \Cutvfs{\vfsv{V}}{\garg{\vfsv{W}}{x}{N}}\\
		& &  && \vfsc{\lt{y}{M}{P}}{x.N} & = & \vfsc{M}{y.\vfsc{P}{x.N}}
	\end{array}
   $$
%
%
   $$
   \begin{array}{ccc}
   \infer{\seqv{\Gamma}{\vfsv{V}}{A}}{\seqC{\Gamma}{V}{A}}
   &
   \infer{\seq{\Gamma}{\vfsc{M}{c}}{B}}{\seqC{\Gamma}{M}{A}&\seqc{\Gamma}{A}{c}{B}}
   &
   \infer{\seq{\Gamma}{\vfst{M}}{A}}{\seqC{\Gamma}{M}{A}}
   \end{array}
   $$
	\caption{The VFS-translation, from $\lbc$ to $\VFS$. Provisos: $(*)$ $P$ is not a value. $(**)$ $Q$ is not a value.}
	\label{tab:vfs-translation}
\end{table*}

\begin{thm}[Simulation]
	\label{thm:simul-vfs} \quad
	\begin{enumerate}
		\item Let $R\in\{B,\ltv,\etalt\}$. If $M\red_R N$ in $\lbc$ then $\vfst{M}\redd\vfst{N}$ in $\VFS$.
		\item Let $R\in\{\ltmn,\ltvn,\assoc\}$. If $M\red_R N$ in $\lbc$ then $\vfst{M}\eq\vfst{N}$ in $\VFS$.
	\end{enumerate}
\end{thm}
%

\textbf{The language $\CPS$.} Recall the CPS-translation of $\lbc$, given in Table \ref{tab:cps-translation}, with target system $\RCPS$, given in Table \ref{tab:RCPS}, our own reworking of Reynold's translation and respective target \cite{DyckhoffLengrandJLC07}. We now introduce a tiny modification in the CPS-translation, an $\eta$-expansion of $k$ in the definition of $\cpsk M$: $\cpsk M  =  \cpsc M {\lb x.kx}$. 
This requires a slight modification of the target system. First, the grammar of commands and continuations becomes:
$$
\begin{array}{rrclcrrcl}
	\text{(Commands)} & M,N & ::= & kV \,|\, KV \,|\, VWK&\qquad&
	\text{(Continuations)} & K & ::= & \lb x.M
\end{array}
$$
The continuation variable $k$ is no longer by itself a continuation -- but nothing is lost with respect to $\RCPS$, since $k$ may be expanded thus:
\begin{equation}\label{eq:expansion-cps}
k \leftarrow_{\etak} \lb x.kx
\end{equation}

Since $K$ is now necessarily a $\lb$-abstraction, the $\etak$-reduction $\lb x.Kx\to K$ of $\RCPS$ becomes a $\sigmav$-reduction in the modified target, and so the latter system has no rule $\etak$.

We do a further modification to the reduction rules: instead of following \cite{SabryWadler97} and having rule $\betav$, we prefer that the modified target system has the rule $(\lb xk.M)WK \to (\lb x.\sub{K}{k}{M})W$, named $\Bv$.
That is, we substitute $K$, but not $W$.\footnote{We could have made this modification in Table \ref{tab:RCPS}, without any change to our results. The only thing to observe is that, if we want $\RCPS$ (or its modification) to consists of syntax that is derivable from the ordinary $\lb$-calculus or Plotkin's call-by-value $\lb$-calculus, then we have to consider these systems equipped with the well-known permutation $(\lb x.M)VV'\red(\lb x.MV')V$.} The new contractum is a $\sigmav$-redex, that can be immediately reduced to produce the effect of $\RCPS$'s rule $\betav$.

In the typed case, the typing rule for $k$ is replaced by this one:
$$
\infer{\seqCPS{k:\neg\ac,\Gamma}{kV}{\perp}}{\seqCPS{\Gamma}{V}{\ac}}
$$
No other modification is introduced w.~r.~t.~Table \ref{tab:RCPS}. The obtained system is named $\CPS$. 

For the modified CPS-translation, we reuse the notation $\cpst{M}$, $\cpsk{M}$, $\cpsv{V}$ and $\cpsc{M}{K}$. From now on, ``CPS-translation'' refers to the modified one, while the original one will be called \uCPS-translation.

In $\RCPS$, $k$ is a fixed continuation variable. In $\CPS$, $k$ is a fixed \emph{covariable}, again occurring exactly once in each command and continuation. The word ``covariable'' intends to be reminiscent of the covariables, or ``names'', of the $\lb\mu$-calculus \cite{ParigotLPAR92}. Accordingly, $kV$ is intended to be reminiscent of the naming constructor of that calculus, and some ``structural substitution'' should be definable in $\CPS$. 

Indeed, consider the following notion of \emph{context} for $\CPS$: $\CC ::= K\ehole \,|\, \ehole WK$. 
Filling the hole $\ehole$ of $\CC$ with $V$ results in the command $\CC\fhole{V}$. Then, we can define the structural substitution operation $\sub{\CC}{k}{-}$ whose critical clause is $\sub{\CC}{k}{(kV)}=\CC\fhole{V}$. 
There is no need to recursively apply the operation to $V$, since $k\notin FV(V)$.

Now in the case $\CC=K\ehole$, the structural substitution $\sub{\CC}{k}{-}$ is the same operation as the ordinary substitution $\sub{K}{k}{-}$, and it turns out that we will only need this case of substitution. That is why we will not see the structural substitution anymore in this paper. 

However, contexts $\CC$ will be crucial for understanding the relationship between $\VFS$ and $\CPS$. In preparation for that, we derive typing rules for contexts of $\CPS$. The corresponding sequents are of the form $\seqCPS{\Gamma|A}{\CC}{\perp}$, where $A$ is the type of the hole of $\CC$. Hence, the command $\CC\fhole{V}$ is typed as follows:
$$
\infer[\CC 1]{\seqCPS{\Gamma}{\CC\fhole{V}}{\perp}}{\seqCPS{\Gamma}{V}{A} & \seqCPS{\Gamma | A}{\CC}{\perp}}
$$
The rules for typing $\CC$ are obtained from the rules for typing $KV$ and $VWK$ in Table \ref{tab:RCPS}, erasing the premise relative to $V$ and declaring $V$'s type as the type of the hole of $\CC$:
$$
\begin{array}{ccc}		
\infer[\CC 2]{\seqCPS{k:\neg\ac,\Gamma|\ac'}{K\ehole}{\perp}}{\seqCPS{k:\neg\ac,\Gamma}{K}{\neg\ac'}}
&\qquad&
\infer[\CC 3]{\seqCPS{k:\neg\ac'',\Gamma|\ac\imp\neg\neg\ac'}{\ehole WK}{\perp}}{\seqCPS{\Gamma}{W}{\ac}&\seqCPS{k:\neg\ac'',\Gamma}{K}{\neg\ac'}}
\end{array}
$$
We also observe that $K_{\CC}:=\lb z.\CC\fhole{z}$ is a continuation, and that $K_{\CC}V\to_{\sigmav}\CC\fhole{V}$ in $\CPS$.

\textbf{VFS vs CPS: the negative translation.} We now see that the CPS-translation can be decomposed as the VFS-translation followed by a \emph{negative} translation of system $VFS$. This latter translation is a CPS-translation, hence involving, at the level of types, the introduction of double negations (hence the name ``negative''). It turns out that this negative translation is an isomorphism between $VFS$ and $\CPS$, at the levels of proofs and proof reduction. This renders the last stage of translation (the negative stage) and its style of representation (the CPS style) an optional addition to what is already achieved with VFS.

The negative translation is found in Table \ref{tab:neg-translation}. It comprises: For each $V\in\VFS$, a value $\ngv{V}$ in $\CPS$; for each $M\in\VFS$, a command $\angt{M}$ and a term $\ngt{M}$ in $\CPS$.

The translation has a typed version, mapping between the typed version of source and target calculi. This requires a translation of types: for each simple type $A$ of $\VFS$, there is an $\ac$-type $\ngv{A}$ and a $\bc$-type $\ngt{A}$, as defined in Table \ref{tab:neg-translation}. The translation preserves typing, according to the admissible rules displayed in the last row of the same table.

\begin{table*}
		\centering
		$$
		\begin{array}{rclcrcl}
			\ngv x & = & x &\qquad&  \angt {(\rt{V})} & = & k\ngv V\\
			\ngv{(\lb x.M)} & = & \lb x. \ngt{M} && \angt{\Cutvfs{V}{x.M}}& = & (\lb x.\angt M)\ngv V \\
			\ngt M & = & \lb k.\angt M && \angt{\Cutvfs{V}{\garg{W}{x}{M}}}& = & \ngv V \ngv W (\lb x.\angt M)
		\end{array}
		$$
		$$
		\ngt{A}=\neg\neg\ngv{A} \qquad \ngv{a}=a \qquad \ngv{(A\imp B)}=\ngv{A}\imp\ngt{B}
		$$
		$$
		\begin{array}{cccc}
			\infer{\seqCPS{\ngv{\Gamma}}{\ngv{V}}{\ngv A}}{\seqv{\Gamma}{V}{A}}
			&
			\infer{\seqCPS{k:\neg\ngv{A},\ngv{\Gamma}}{\angt{M}}{\perp}}{\seq{\Gamma}{M}{A}}
			&
			\infer{\seqCPS{\ngv{\Gamma}}{\ngt{M}}{\ngt{A}}}{\seq{\Gamma}{M}{A}}
		\end{array}
		$$
	\caption{The negative translation, from $\VFS$ to $\CPS$, with admissible typing rules}
	\label{tab:neg-translation}
\end{table*}
The negative translation is defined at the level of terms and values. How about formal contexts? A formal context $c$ is translated as a context $\angc{c}$ of $\CPS$, defined as follows:
$$
\begin{array}{rclcrcl}
	\angc{(x.M)} &=& (\lb x.\angt{M})\ehole&\qquad\qquad&
	\angc{\garg{W}{x}{M}} &=& \ehole\ngv{W}(\lb x.\angt{M})
\end{array}
$$
Then the definition of $\angt{\Cutvfs{V}{c}}$ can be made uniform in $c$ as $\angc{c}\fhole{\ngv{V}}$. The translation of non-values $\angt{\Cutvfs{V}{c}}$ is thus defined as filling the (translation) of $V$ in the hole of the actual context $\angc{c}$ that translates the formal context $c$. Hence the name ``value-filling'' of the translation.

We have two admissible typing rules:
$$
\infer[(a)]{\seqCPS{k:\neg\ngv{B},\ngv{\Gamma}|\ngv{A}}{\angc{c}}{\perp}}{\seqc{\Gamma}{A}{c}{B}}\qquad \infer[(b)]{\seqCPS{k:\neg\ngv{B},\ngv{\Gamma}}{K_{\angc{c}}}{\neg\ngv{A}}}{\seqc{\Gamma}{A}{c}{B}}
$$
Rule (a) follows from typing rules $\CC 2$ and $\CC 3$; rule (b) is obtained from (a) and rule $\CC 1$.

It is no exaggeration to say that typing rule (b) is the heart of the negative translation. In the sequent calculus $\VFS$ we can single out a formula $A$ in the l.~h.~s.~of the sequent to act as the type of the hole of a (formal) context $c$. In $\CPS$, we have the related concept of a continuation $K$, a function of type $A\imp\perp$. The type $B$ of $c$ has to be stored as the negated type $\neg B$ of a special variable $k$. Cutting with $c$ in the sequent calculus corresponds to applying $K$, to obtain a command, of type $\perp$. But the cut produces a term of type $B$, while the best we can do in $\CPS$ is to abstract $k$, to obtain $\neg\neg B$. In the sequent calculus, a type $A$ may have uses in both sides of the sequent. To approximate this flexibility in $\CPS$, a type $A$ requires types $\ac$, $\neg\ac$, and $\neg\neg\ac=\bc$, presupposing $\perp$.


\begin{thm}[Decomposition of the CPS-translation]\label{thm:decomposition} \qquad
\begin{enumerate}
	\item For all $V\in\lbc$, $\ngv{{\vfsv{V}}}=\cpsv{V}$.
	\item For all $M\in\lbc$, $N\in\VFS$, $\angt{\vfsc{M}{x.N}}=\cpsc{M}{\lb x.\angt{N}}$.
	\item For all $M\in\lbc$, $\angt{{\vfst{M}}}=\cpsk{M}$.
	\item For all $M\in\lbc$, $\ngt{{\vfst{M}}}=\cpst{M}$.
\end{enumerate}	
\end{thm}

Nothing is lost, if we wish to replace $\CPS$ with $\VFS$, because the negative translation is an isomorphism. Its inverse translation comprises: For each term $P\in\CPS$, a term $\post{P}\in\VFS$; for each command $M\in\CPS$, a term $\apos{M}\in\VFS$; for each value $V\in\CPS$, a value $\posv{V}\in\VFS$. 
The definition is as follows:
$$
\begin{array}{rcl}
\post{(\lb k.M)}&=&\apos M\\
\apos{(kV)}&=&\rt{(\posv{V})}\\
\apos{((\lb x.M)V)}&=&\Cutvfs{\posv{V}}{x.\apos{M}}\\
\apos{(VW(\lb x.M))}&=&\Cutvfs{\posv{V}}{\garg{\posv{W}}{x}{\apos{M}}}\\
\posv{x}&=&x\\
\posv{(\lb x.P)}&=&\lb x.\post{P}
\end{array}
$$
\begin{thm}[$\VFS\cong\CPS$]\label{thm:iso} \quad
	\begin{enumerate}
		\item For all $M,V\in\VFS$, $\post{{\ngt{M}}}=M$ and $\apos{{\angt{M}}}=M$ and $\posv{{\ngv{V}}}=V$.
		\item For all $P,M,V\in\CPS$, $\ngt{{\post{P}}}=P$ and $\angt{{\apos{M}}}=M$ and $\ngv{{\posv{V}}}=V$.
		\item If $M_1\to M_2$ in $\VFS$ then $\angt{M_1} \to\angt{M_2}$ in $\CPS$ (hence $\ngt{M_1}\to\ngt{M_2}$ in $\CPS$).
		\item If $M_1\to M_2$ in $\CPS$ then $\apos{M_1}\to\apos{M_2}$ in $\VFS$. Hence If $P_1\to P_2$ in $\CPS$ then $\post{P_1}\to\post{P_2}$ in $\VFS$.
	\end{enumerate}
\end{thm}


%% file: back-direct-style.tex
\section{Back to direct style}\label{sec:bds}

We now do to the VFS-translation what \cite{FlanaganSabryDubaFelleisen93,SabryWadler97} did to the CPS-translation, that is, try to find a program transformation in the source language $\lbc$ that corresponds to the effect of the translation. We have seen in Section \ref{sec:VFS} that the VFS-translation identifies reduction steps generated by $\ltmn$, $\ltvn$ and $assoc$. So we start from the normal forms w.~r.~t.~these rules, that is, from the kernel $\ANF$ (recall Table \ref{tab:kernel-lC}). We first identify two sub-syntaxes relevant in this analysis. Next, we point out the proof-theoretical meaning of such alternative. 

\textbf{Two sub-kernels of $\ANF$.} It turns out that the syntax of $\ANF$, despite its simplicity, still contains several dilemmas: (1) Do we need a let-expression whose actual parameter is a value $V$? Or should we normalize with respect to $\ltv$? (2) Do we need $VW$ to stand alone as a term and also as the actual parameter of a let-expression? (3) Is $\etalt$ a reduction or an expansion? 
Some of these dilemmas give rise to the following diagram:
\begin{equation}\label{eq:expansion-lbc}
		\xymatrix{
			VW & \ar@<0ex>[l]_{\ltv}  \lt{x}{V}{xW}\\
			\lt{y}{VW}{y} \ar@<0ex>[u]^{\etalt}  &\ar@<0ex>[l]^{\!\!\!\!\!\!\! \ltv}  \ar@<0ex>[u]_{\etalt} 
			\lt{x}{V}{\underbrace{\lt{y}{xW}{y}}_{c_x}}}
\end{equation}

We take this diagram as giving, in its lower row, two different ways of expanding $VW$. These two alternatives signal two sub-syntaxes of $\ANF$ without $VW$. In the alternative corresponding to the expansion $\lt{y}{VW}{y}$, we are free to, additionally, normalize w.~r.~t.~$\ltv$ and get rid of the form $\lt{x}{V}{M}$. In the alternative $\lt{x}{V}{\lt{y}{xW}{y}}$, we are not free to normalize w.~r.~t.~$\ltv$, as otherwise we might reverse the intended expansions. 
In both cases, values are $V,W::= x\,|\,\lb x.M$. Moreover, we do not want to consider $\etalt$ as a reduction rule; and rule $\Bv'$ disappears, since there are no applications $VW$. 

In the first sub-kernel, named $\ces$, terms $M$ are given by the grammar
$$ M \, ::= \,  V \, |\, \lt{x}{VW}{M} \enspace.$$
We call this representation \emph{continuation enclosing} style, since the ``serious'' (=non-value) terms have the form of an application $VW$ enclosed in a let-expression. The unique reduction rule of $\ces$ is
$$
(\betav) \qquad \lt{y}{(\lb x.M)V}{P}  \red  \ltc{y}{\sub{V}{x}{M}}{P}
$$
In $\ANF$, it corresponds to a $\Bv$-step followed by $\ltv$-step. The operation $\ltc{y}{M}{P}$ of $\ANF$ is reused, except that the base case of its definition integrates a further $\ltv$-step: $\ltc{y}{V}{P}=\sub{V}{y}{P}$.

In the second sub-kernel, named $\VES$, terms are given by the grammar
$$
\begin{array}{rcl}
M,N	& ::= & V \,|\, \lt{x}{V}{c_x}\\
c_x &::=  & M\,|\,\lt{y}{xW}{N}, \text{ where $x\notin FV(W)\cup FV(N)$}
\end{array}
$$
We call this representation \emph{value enclosed} style, since the serious terms have the form of a value enclosed in a let-expression. There are two reduction rules:
$$
\begin{array}{rrcl}
(\Bv)& \lt{y}{(\lb x.M)}{\lt{z}{yV}{P}}  & \red & \lt{x}{V}{\ltc{z}{M}{P}}\\
(\ltv) & \lt{y}{V}{N} & \red & \sub{V}{y}{N}
\end{array}
$$

In $\VES$, we define $\ltc{y}{M}{P}$ and $\ltc{y}{c_z}{P}$, which are a term and an element of the class $c_z$, respectively, the latter satisfying $z\notin FV(P)$. The definition is by simultaneous recursion on $M$ and $c_z$ as follows:
$$
\begin{array}{rcl}
	\ltc{y}{V}{P} &=& \lt{y}{V}{P}\\
	\ltc{y}{(\lt{z}{V}{c_z})}{P} &=& \lt{z}{V}{\ltc{y}{c_z}{P}}\\
		\ltc{y}{(\lt{x}{zW}{N})}{P} &=& \lt{x}{zW}{\ltc{y}{N}{P}}
	\end{array}
$$
In the second equation, since in the l.~h.~s.~$P$ is not in the scope of the (inner) let-expression, we may assume $z\notin FV(P)$. So, the proviso for the call $\ltc{y}{c_z}{P}$ in the r.~h.~s.~is satisfied. In the third equation, $c_z$ in the l.~h.~s.~is $\lt{x}{zW}{N}$. By definition of $c_z$, $z\notin FV(W)\cup FV(N)$; moreover, we may assume $z\notin FV(P)$: hence the r.~h.~s.~is in $c_z$. 

Despite the trouble with variable conditions, this definition corresponds to the operator $\ltc{y}{M}{P}$ of $\ANF$ restricted to the syntax of $\VES$. Therefore, rule $\Bv$ of $\VES$ corresponds, in $\ANF$, to a $\ltv$-step followed by a $\Bv$-step.


\textbf{Proof-theoretical alternative.} We now see that $\VES$ is related to the sequent calculus $\VFS$, while $\ces$ is related to a fragment $\cnf$ of the call-by-value $\lb$-calculus with generalized applications $\lbjv$ introduced in \cite{jesCSL20}. In both cases, the relation is an isomorphism, in the sense of a type-preserving bijection with a 1-1 simulation of reduction steps.

\begin{thm}\label{thm:two-isomorphisms}$\VES\cong\VFS$ and 	$\ces\cong\cnf$.
\end{thm}

Therefore the alternative between the two sub-kernels corresponds to the alternative between two proof-systems for call-by-value, the sequent calculus $\LJQ$ and the natural deduction system with general elimination rules behind $\lbjv$.

A $\lbjv$-term is either a value or a generalized applications $\gapp{M}{N}{x}{P}$, with typing rule
$$
\infer{\seqJ{\Gamma}{\gapp{M}{N}{x}{P}}{C}}{\seqJ{\Gamma}{M}{A\imp B}&\seqJ{\Gamma}{N}{A}&\seqJ{\Gamma,x:B}{P}{C}}
$$
If the head term $M$ is itself an application $\gapp{M_1}{M_2}{y}{M_3}$, then $M_3$ has type $A\imp B$ and the term can be rearranged as $\gapp{M_1}{M_2}{y}{\gapp{M_3}{N}{x}{P}}$, to bring $M_3$ and $N$ together. This is a known \emph{commutative conversion} \cite{JoachimskiMatthesRTA2000}, here named $\pi_1$, which aims to convert the head term $M$ to a value $V$. On the other hand, if the argument $N$ is itself an application $\gapp{N_1}{N_2}{y}{N_3}$, then $N_3$ has type $A$ and the term can be rearranged as $\gapp{N_1}{N_2}{y}{\gapp{M}{N_3}{x}{P}}$, to bring $M$ and $N_3$ together. This is a conversion $\pi_2$ which has \emph{not} been studied, and which aims to convert the argument $N$ to a value $W$.

The combined effect of $\pi:=\pi_1\cup\pi_2$ is to reduce generalized applications to the form $\gapp{V}{W}{x}{P}$, called \emph{commutative normal form}. On these forms, the $\betav$-rule of $\lbjv$ reads 
$$(\betav)\qquad\qquad \gapp{(\lb y.M)}{W}{x}{P}\to \lsub{\sub{W}{y}{M}}{x}{P}$$
The \emph{left substitution} operation $\lsub{N}{x}{P}$ is defined by 
$$
\lsub{V}{x}{P}=\sub{V}{x}{P} \qquad\qquad \lsub{\gapp{V}{W}{y}{N_3}}{x}{P}=\gapp{V}{W}{y}{\lsub{N_3}{x}{P}}
$$
The commutative normal forms, equipped with $\betav$, constitute the system $\cnf$.

\begin{table}[t]
	$$
	\begin{array}{rcl}
		\Psi(V) & = & \rt{\Psi_v(V)}\\
		\Psi (\lt{x}{V}{c_x})& = & \Cutvfs{\Psi_v V}{\Psi_x(c_x)} \\
		\Psi_v(x) &=& x \\
		\Psi_v(\lb x.M) & =& \lb x. \Psi M   \\
		\Psi_x(M) &=& x. \Psi M\\
		\Psi_x(\lt{y}{xW}{N}) & =& \garg{\Psi W}{y}{\Psi N}\\
		\\
		\Theta(\rt{V}) &=& \Theta_v(V)\\
		\Theta(\Cutvfs{V}{c}) &=& \lt{x}{\Theta_v V}{\Theta_x(c)}\\
		\Theta_v(x) &=& x \\
		\Theta_v(\lb x.M) & =& \lb x. \Theta M  \\
		\Theta_x(y.M) &=& \sub{x}{y}{(\Theta M)}\\
		\Theta_x \garg{W}{y}{N} & =& \lt{y}{x(\Theta_v W)}{\Theta N}
	\end{array}
	$$
	\caption{Translation from $\VES$ to $\VFS$ and vice-versa.}
	\label{table:ves-vfs}
\end{table}

\begin{table}[t]
	$$
	\begin{array}{rcl}
		\Upsilon (x) &=& x\\
		\Upsilon (\lb x.M) & = & \lb x. \Upsilon M\\
		\Upsilon(\lt{x}{VW}{M}) & = & \Upsilon V \garg{\Upsilon W}{x}{\Upsilon M}\\
		\\
		\Phi (x) &=& x\\
		\Phi (\lb x.M) & = & \lb x. \Phi M\\
		\Phi(\li{V}{W}{x}{M}) & = & \lt{x}{\Phi V \Phi W}{\Phi M}
	\end{array}
	$$
	\caption{Translation from $\ces$ to $\cnf$ and vice-versa.}
	\label{table:ces-cnf}
\end{table}

The announced isomorphisms are given in Tables \ref{table:ves-vfs} and \ref{table:ces-cnf}. 
The map $\Psi:\VES\to\VFS$ 
requires the key auxiliary map $\Psi_x$, whose design is guided by types: if $\seqC{\Gamma,x:A}{c_x}{B}$ then $\seqc{\Gamma}{A}{\Psi_x(c_x)}B$. 
The isomorphism $\Upsilon:\ces\to\cnf$ should be obvious. 
It can be proved that the operation $\ltc{y}{M}{P}$ in $\ces$ is translated as left substitution: $\Upsilon(\ltc{y}{M}{P})=\lsub{\Upsilon M}{y}{\Upsilon P}$.

A final point. The sub-kernel $\VES$ is isomorphic to the CPS-target, after composition with the negative translation: $\VES\cong\VFS\cong\CPS$. 
A variant of the negative translation delivers:
\begin{thm}\label{thm:iso-cnf} $\cnf\cong\cps$.
\end{thm}

So we also have $\ces\cong\cnf\cong\cps$. Here $\cps$ is the sub-calculus of $\CPS$ where commands $KV$ are omitted and $\sigmav$ normalization is enforced. Its unique reduction rule, named $\betav$, becomes 
$$
(\betav)\qquad\qquad (\lb y.\lb k.M)W(\lb x.N)\to\sub{\lb x.N}{k}{\sub{W}{y}{M}}
$$
The definition of substitution $\sub{\lb x.N}{k}{M}$ has the following critical clause: 
$$
\sub{\lb x.N}{k}{(kV)}=\sub{V}{x}{N}
$$
This clause does the reduction of the $\sigmav$-redex $(\lb x.N)V$ on the fly; and it echoes the critical clause 
of a structural substitution. Moreover, $\cps$ is the target of a version of the CPS-translation, obtained by changing just one clause: $\cpsc{V}{\lb x.M}=\sub{\cpsv{V}}{x}{\cpst{M}}$.

The variant of the negative translation yielding $\cnf\cong\cps$ is defined by 
$$
\angt{(\gapp{V}{W}{x}{M})}=\ngv{V}\ngv{W}(\lb x.\angt{M})
$$
All the other needed clauses as before. For the isomorphism, we have to prove: 
$$
\angt{(\lsub{N}{x}{M})}=\sub{\lb x.\angt{M}}{k}{\angt{N}}
$$
This is a last minute bonus: a $\cps$ explanation of left substitution.

%% file: conclusions.tex
\section{Conclusions}\label{sec:conc}


\textbf{Contributions.} %
%
We list our main contribution: the VFS-translation; the negative translation as an isomorphism between the VFS and CPS targets; the decomposition of the CPS-translation in terms of the VFS-translation and the negative translation; the two sub-kernels of $\lbc$ and their perfect relationship with appropriate fragments of the sequent calculus $\LJQ$ and natural deduction with general eliminations; the reworking of the term calculus for $\LJQ$.

In all, we took the polished account of the essence of CPS, obtained in \cite{SabryWadler97} and illustrated in Fig.~\ref{fig:essence}, and revealed a rich proof-theoretical background, as in Fig.~\ref{fig:logical-essence}, with a double layer of sub-kernels, under a layer of expansions (see the dotted lines in Fig.~\ref{fig:logical-essence} and recall (\ref{eq:expansion-vfs}), (\ref{eq:expansion-cps}), and (\ref{eq:expansion-lbc})), intersecting an intermediate zone, between the source language and the CPS targets, of calculi corresponding to proof systems.

\textbf{Related work.} In \cite{DyckhoffLengrandJLC07}, $LJQ$ is studied as a source language, while the CPS translation of $\LJQ$ is a tool to establish indirectly a connection with $\lbc$, through their respective kernels, in order to confirm that cut-elimination in $\LJQ$ is connected with call-by-value computation. There is nothing wrong with using the sequent calculus as source language and translating it with CPS: this has been done abundantly, even by the first author \cite{CurienHerbelinICFP00,WadlerICFP2003,DyckhoffLengrandJLC07,jesRalphLuisLMCS2009}. But the point made here is that the sequent calculus should also be used as a tool to analyze the CPS-translation, and is able to play a special role as an intermediate language.

The sequent calculus was put forward as an intermediate representation for compilation of functional programs in \cite{DownenMaurerAriolaJones2016}. This study addresses compilation of programs for a real-world language; designs an intermediate language \emph{Sequent Core} (SC) inspired in the sequent calculus for such source language; and compares SC with CPS heuristically w.~r.~t.~several desirable properties in the context of optimized compilation. In the present paper, we address the foundations of compilation, employing theoretical languages; pick the sequent calculus $\LJQ$, which is a standard systems with decades of history in proof-theory \cite{DyckhoffLengrandCiE2006}; and compare $LJQ$ and CPS, not through a benchmarking of competing languages, but through mathematical results showing their intimate connection. 

\textbf{Future work.} We know an appropriate CPS target will be capable of interpreting a classical extension of our chosen source language. The problem in moving in this direction is that there is no standard extension of $\lbc$ with control operators readily available. Source languages with let-expressions and control operators can be found in \cite{HerbelinZimmermann09,jesAPAL2013}, but adopting them means to redo all that we have done here -- that is another project. On the other hand, maybe a system with generalized applications will make a good source language. The system $\lbjv$ performed well in this paper, since its sub-kernel of administrative normal forms ($\cnf$) is reachable without consideration of expansions -- a sign of a well calibrated syntax.

%% file: originalLJQ.tex
\section{The original LJQ system}\label{sec:original}

The original calculus by Dyckhoff-Lengrand is recalled in Table \ref{tab:original-lQ}.

\begin{table}[t]
	
		$$\begin{array}{r r c l}
			\text{(terms)} & M, N &::= &\rt{V}  \, | \,  \li{x}{V}{y}{N} \, | \, \CVN{V}{x}{N} \, | \,  \CMN{M}{x}{N} \\
			\text{(values)} & V,W &::=& x \, | \, \lb x.M \, | \, \CVW{V}{x}{W}
		\end{array}$$
		
		$$ \begin{array}{rrcll}
			(1)	& \CMN{\rt{(\lb x.M)}}{y}{\li{y}{V}{z}{N}} &\to&  \CMN{\CMN{\rt{V}}{x}{M}}{z}{N} & (a) \\
			(2)	& \CMN{\rt{x}}{y}{N} &\to & \sub{x}{y}{N}  \\
			(3)	& \CMN{M}{x}{\rt{x}} &\to & M   \\
			(4)	& \CMN{\li{z}{V}{y}{P}}{x}{N} &\to & \li{z}{V}{y}{\CMN{P}{x}{N}} \\
			(5)	& \CMN{\CMN{\rt{W}}{y}{\li{y}{V}{z}{P}}}{x}{N} &\to & \CMN{\rt{W}}{y}{\li{y}{V}{z}{\CMN{P}{x}{N}}} & (b) \\
			(6)	& \CMN{\CMN{M}{y}{P}}{x}{N} &\to &\CMN{M}{y}{\CMN{P}{x}{N}} & (c)\\
			(7)	& \CMN{\rt({\lb x.M})}{y}{N} &\to & \CVN{\lb x.M}{y}{N} & (d)\\
			(8)	& \CVW{V}{x}{x} &\to & V \\
			(9)	& \CVW{V}{x}{y} &\to & y & (e)\\
			(10)	& \CVW{V}{x}{(\lb y.M)} &\to & \lb y.\CVN{V}{x}{M} \\
			(11)	& \CVN{V}{x}{\rt{W}} &\to & \rt{(\CVW{V}{x}{W})} \\
			(12)	& \CVN{V}{x}{\li{x}{W}{z}{N}} &\to & \CVN{\rt{V}}{x}{\li{x}{\CVW{V}{x}{W}}{z}{\CVN{V}{x}{N}}} \\
			(13)	& \CVN{V}{x}{\li{y}{W}{z}{N}} &\to & \li{y}{\CVW{V}{x}{W}}{z}{\CVN{V}{x}{N}} & (e)\\
			(14)	& \CVN{V}{x}{\CMN{M}{y}{N}} &\to &  \CMN{\CVN{V}{x}{M}}{y}{\CVN{V}{x}{N}}\\
		\end{array}$$
		
		Provisos: $(a)$ $y \notin FV(V) \cup FV(N)$. $(b)$ $y \notin FV(V) \cup FV(P))$. $(c)$ If rule (5) does not apply. $(d)$ If rule (1) does not apply. $(e)$ $x\neq y$.
		
		$$\begin{array}{c c}
			\infer[Ax]{\seqv{\Gamma, x:A}{x}{A}}{} & \infer[Der]{\seq{\Gamma}{\rt{V}}{A}}{\seqv{\Gamma}{V}{A}}\\
			\\
			\infer[R\!\imp]{\seqv{\Gamma}{\lb x.M}{A \imp B}}{\seq{\Gamma, x:A}{M}{B}} &  \infer[Cut_3]{\seq{\Gamma}{\CMN{M}{x}{N}}{B}}{\seq{\Gamma}{M}{A} & \seq{\Gamma, x:A}{N}{B}}\\
			\\
			\infer[Cut_1]{\seqv{\Gamma}{\CVW{V}{x}{W}}{B}}{\seqv{\Gamma}{V}{A} & \seqv{\Gamma, x:A}{W}{B}} & \infer[Cut_2]{\seq{\Gamma}{\CVN{V}{x}{N}}{B}}{\seqv{\Gamma}{V}{A} & \seq{\Gamma, x:A}{N}{B}}\\
			\\
			\multicolumn{2}{c}{\infer[L\!\imp]{\seq{\Gamma, x: A \imp B}{\li{x}{V}{y}{N}}{C}}{\seqv{\Gamma, x:A \imp B}{V}{A} & \seq{\Gamma, x:A \imp B, y:B}{N}{C}}} 
		\end{array}$$
		
	\caption{The original calculus by Dyckhoff-Lengrand}
	\label{tab:original-lQ}
\end{table}

%% file: kernel-lbc.tex
\section{Kernel of $\lbc$}\label{sec:kernel-lbc}

Our presentation of the kernel of $\lbc$ given in Table \ref{tab:kernel-lC} is very close to the original one in \cite{SabryWadler97}, as we now see. In \cite{SabryWadler97}, the terms $M$ of the kernel are generated by the grammar:
$$
\begin{array}{rcl}
	M,N,P & ::= & \K\fhole{V} | \K\fhole{VW}\\
	V,W & ::= & x | \lb x.M \\
	\K & ::= & \ehole | \lt{x}{\ehole}{P}
\end{array}
$$
We take for granted the sets of terms and values of $\lbc$, together with the set of contexts of $\lbc$, which are $\lbc$-terms with a single hole, and the concept of hole filling in such contexts. This grammar defines simultaneously a subset of the terms of $\lbc$, a subset of the values of $\lbc$, and a subset of the contexts of $\lbc$.

The second production in the grammar of terms, $\K\fhole{VW}$, should be understood thus: given in the kernel values $V$, $W$ and a context $\K$, the $\lbc$-term $\K\fhole{VW}$, obtained by filling the hole of $\K$ with the $\lbc$-term $VW$, is in the kernel. In $\lbc$, $VW$ is a subterm of $\K\fhole{VW}$; but, as we observed in Section \ref{sec:back}, in the kernel, the term $VW$ is not an immediate subterm of $\K\fhole{VW}$ -- the immediate subexpressions are just $V$, $W$, and $\K$. Notice the $\lbc$-term $M=VW$ is a term in the kernel, generated by the second production of the grammar with $\K=\ehole$. But that second production \emph{should not} be interpreted as $\K\fhole{M}$ with $M=VW$.

There is no primitive $\K\fhole{M}$ in the kernel. Instead, there is the operation $(M:\K)$, defined by recursion on $M$ as follows:
$$
\begin{array}{rcl}
	(V:\K) &=& \K\fhole{V}\\
	(VW:\K) &=& \K\fhole{VW}\\
	(\lt{x}{V}{M}:\K) &=&\lt{x}{V}{(M:\K)}\\
	(\lt{x}{VW}{M}:\K) &=&\lt{x}{VW}{(M:\K)}
\end{array}
$$
It is easy to see that $(M:\lt{x}{\ehole}{P})=\ltc{x}{M}{P}$ and that $(M:\ehole)=M$.

In \cite{SabryWadler97}, the kernel has the following reduction rule
$$
(\beta.v) \qquad \K\fhole{(\lb x.M)V}\to(\sub{V}{x}{M}:\K) \enspace.
$$
There is no need for the requirement of maximal $\K$ in this rule, as done in \cite{SabryWadler97}, once the above clarification about $\K\fhole{VW}$ is obtained. We now see the relationship between $\beta.v$ and our $\Bv$ and $\Bv'$.

Let $\K=\lt{y}{\ehole}{P}$. Then rule $\Bv$ can re written as 
$$\K\fhole{(\lb x.M)V}\to\lt{x}{V}{(M:\K)} \enspace.$$
The contractum is a $\ltv$-redex, which could be immediately reduced, to achieve the effect of $\beta.v$. Here we prefer to delay this $\ltv$-step, and the same applies to our rule $\Bv'$, which corresponds to the case $\K=\ehole$. This issue of delaying $\ltv$ is also seen in Section \ref{sec:bds}.

Finally, rule $\etalt$ in \cite{SabryWadler97} reads $\lt{x}{\ehole}{\K\fhole{x}}\to\K$. We argue that in our presentation we can derive $$(M:\lt{x}{\ehole}{\K\fhole{x}})\to(M:\K)\enspace.$$

If $\K=\ehole$, then we have to prove $\ltc{x}{M}{x}\to M$. This is proved by an easy induction on $M$: the case $M=V$ (resp.~$M=VW$) gives rise to a $\sigmav$-step (resp.~$\etalt$-step); the remaining two cases follow by induction hypothesis. 

If $\K=\lt{y}{\ehole}{P}$, then we have to prove $\ltc{x}{M}{\lt{y}{x}{P}}\to\ltc{y}{M}{P}$. Now $\lt{y}{x}{P}\to_{\ltv}\sub{y}{x}{P}$. Since $Q\to Q'$ implies $\ltc{x}{M}{Q}\to\ltc{x}{M}{Q'}$, we obtain $\ltc{x}{M}{\lt{y}{x}{P}}\to\ltc{x}{M}{\sub{y}{x}{P}}=_{\alpha}\ltc{y}{M}{P}$.

%% file: proofs.tex
\section{Proofs}\label{sec:proofs}


\subsection{Proofs of Section \ref{sec:LJQ}}

We reuse the maps $(\_)^{\sharp}:\lbc\to\lbqo$ and $(\_)^{\flat}:\lbqo\to\lbc$ by Dyckhoff-Lengrand \cite{DyckhoffLengrandJLC07}.


\begin{lem}[Equational correspondence between $\lbc$ and $\lbqmo$]\label{lem:ec-lc-lbmo} \qquad
\begin{enumerate}
	 \item If $M\red N$ in $\lbc$ then $M^{\sharp}\redd N^{\sharp}$ in $\lbqmo$. 
	 \item If $M\red N$ in $\lbqmo$ then $M^{\flat}\redd N^{\flat}$ in $\lbc$. 
	 \item $M\redd M^{\flat\sharp}$ in $\lbqmo$. 
	 \item $M\conv M^{\sharp\flat}$ in $\lbc$.
\end{enumerate}
\end{lem}
	\begin{proof} Dyckhoff-Lengrand proved there is an equational correspondence between $\lbc$ and $\lbqo$: 
		\begin{itemize}
			\item[(a)] If $M\red N$ in $\lbc$ then $M^{\sharp}\redd N^{\sharp}$ in $\lbqo$.
			\item[(b)] If $M\red N$ in $\lbqo$ then $M^{\flat}\redd N^{\flat}$ in $\lbc$.
			\item[(c)] $M\redd M^{\flat\sharp}$ in $\lbqo$.
			\item[(d)] $M\conv M^{\sharp\flat}$ in $\lbc$.
		\end{itemize}
		Items (a) - (d) entail items (1) to (4). Items (1) and (3) follow from items (a) and (c), respectively, because $\redd$ of $\lbqo$ is contained in $\redd$ of $\lbqmo$. Item (4) is item (d). As to item (2), the only concern is when $M\to N$ by rule (6), due to the waiving of the proviso. But an inspection of the proof of (b) in \cite{DyckhoffLengrandJLC07} shows the proviso plays no role (because a step of the form (6), regardless of the proviso, is identified by the refined Fischer CPS translation of $\lbqo$, and therefore identified by $(\_)^{\flat}$).
	\end{proof}

This is also a pre-Galois connection from $\lbqmo$ to $\lbc$. Because of this, $\lbqmo$ inherits confluence of $\lbc$, as $\lbqo$ did.

\begin{cor}
	$\lbqmo$ is confluent.
\end{cor}

Next we prove 3 Lemmas needed to obtain Theorem \ref{thm:reflection-of-simplified}.

\begin{lem}\label{lem:subs}
	For all $V,W,M \in \lbq$:
	\begin{enumerate}
		\item $\CVW{V}{x}{W} \redd \sub{V}{x}{W}$ in $\lbqmo$.
		\item $\CVN{V}{x}{M} \redd \sub{V}{x}{M}$ in $\lbqmo$.
	\end{enumerate} 
\end{lem}

\begin{proof}
	By simultaneous induction on $W$ and $M$.
	
		Case $W=x$.		
		$$\begin{array}{rclll}
			\CVW{V}{x}{x} & \red_8 & V & \eq & \sub{V}{x}{x}
		\end{array}$$
	
		Case $W=y$.
		$$\begin{array}{rclll}
			\CVW{V}{x}{y} & \red_9 & y & \eq & \sub{V}{x}{y}
		\end{array}$$
	
		Case $W=\lb y. P$.
		$$\begin{array}{rcll}
			\CVW{V}{x}{(\lb y. P)} & \red_{10} & \lb x. \CVN{V}{x}{P}\\
			& \redd & \lb y. \sub{V}{x}{P} & \text{(by IH)}\\
			& \eq & \sub{V}{x}{(\lb y.P)}
		\end{array}$$
	
		Case $M=\rt{W}$.			
		$$\begin{array}{rcll}
			\CVN{V}{x}{\rt{W}} & \red_{11} & \rt{(\CVW{V}{x}{W})}\\
			& \redd & \rt{(\sub{V}{x}{W})} & \text{(by IH)}\\
			& \eq & \sub{V}{x}{(\rt{W})}
		\end{array}$$
	
		The remaining cases are analogous to the previous one.
	
\end{proof}

\begin{lem}\label{lem:sub-of-sub}
	For all $V, W, M \in \lbq$:
	\begin{enumerate}
		\item $\sub{y}{x}{(\sub{W}{z}{V})}\eq\sub{\sub{y}{x}{W}}{z}{(\sub{y}{x}{V})}$.
		\item $\sub{y}{x}{(\sub{W}{z}{M})}\eq \sub{\sub{y}{x}{W}}{z}{(\sub{y}{x}{M})}$.
	\end{enumerate}
	 
\end{lem}

\begin{proof}
	By simultaneous induction on $V$ and $M$.
	
	Case $V=x$.			
	$$\begin{array}{rclll}
		\sub{\sub{y}{x}{W}}{z}{(\sub{y}{x}{x})} & \eq & \sub{\sub{y}{x}{W}}{z}{y}\\
		& \eq &  y \\
		& \eq & \sub{y}{x}{x} \\
		& \eq & \sub{y}{x}{(\sub{W}{z}{x})}
	\end{array}$$

	Case $V=z$.
	$$\begin{array}{rclll}
		\sub{\sub{y}{x}{W}}{z}{(\sub{y}{x}{z})} & \eq & \sub{\sub{y}{x}{W}}{z}{z}\\
		& \eq &  \sub{y}{x}{W} \\
		& \eq &  \sub{y}{x}{(\sub{W}{z}{z})}
	\end{array}$$

	Case $V=u$, $x \neq u \neq z$.
	$$\begin{array}{rclll}
	\sub{\sub{y}{x}{W}}{z}{(\sub{y}{x}{u})} & \eq & \sub{\sub{y}{x}{W}}{z}{u}\\
		& \eq &  u \\
		& \eq & \sub{y}{x}{u}\\
		& \eq &  \sub{y}{x}{(\sub{W}{z}{u})}
	\end{array}$$

	Case $V=\lb u.P$.
	$$\begin{array}{rclll}
		\sub{\sub{y}{x}{W}}{z}{(\sub{y}{x}{(\lb u.P)})} & \eq & \lb u. \sub{\sub{y}{x}{W}}{z}{(\sub{y}{x}{P})}\\
		& \eq & \sub{y}{x}{(\sub{W}{z}{P})} & \text{(by IH)}\\
		& \eq &  \sub{y}{x}{(\sub{W}{z}{(\lb u.P)})}
	\end{array}$$
	
	The remaining cases are analogous.
	
\end{proof}

The map $\smp{(\_)}:\lbqmo\to\lbq$, based on the idea of translating the omitted cuts by calls to substitution, is defined as follows:

$$\begin{array}{rclcrcl}
	\smp{(\rt{V})} & \eq & \rt{\smpv{V}}&&\smpv{x} & \eq & x\\
	\smp{\li{x}{V}{y}{N}} & \eq & \li{x}{\smpv{V}}{y}{\smp{N}}&&\smpv{(\lb x.M)} & \eq & \lb x. \smp{M}\\
	\smp{\CVN{V}{y}{N}} & \eq & \sub{\smpv{V}}{x}{\smp{N}}&&\smpv{\CVW{V}{x}{W}} & \eq & \sub{\smpv{V}}{x}{\smpv{W}}\\
	\smp{\CMN{M}{x}{N}} & \eq & \Cut{\smp{M}}{x}{\smp{N}}
\end{array}$$

\begin{lem}\label{lem:maps-of-subs}
	For all $V,M \in \lbqmo$:
	\begin{enumerate}
		\item $\smpv{(\sub{y}{x}{V})} \eq \sub{\smpv{y}}{x}{\smpv{V}}$.
		\item $\smp{(\sub{y}{x}{M})} \eq \sub{\smpv{y}}{x}{\smp{M}}$.
	\end{enumerate}
\end{lem}

\begin{proof}
	By simultaneous induction on $V$ and $M$.
	
	Case $V=x$.
		$$\begin{array}{rclll}
			\smpv{(\sub{y}{x}{x})} & \eq & \smpv{y}\\
			& \eq & \sub{\smpv{y}}{x}{x}\\
			& \eq & \sub{\smpv{y}}{x}{\smpv{x}}
		\end{array}$$
		
		Case $V=z$.
		$$\begin{array}{rclll}
			\smpv{(\sub{y}{x}{z})} & \eq & \smpv{z} \\
			& \eq & z \\
			& \eq & \sub{\smpv{y}}{x}{z}\\
			& \eq & \sub{\smpv{y}}{x}{\smpv{z}}
		\end{array}$$
	
	 	Case $V=\lb z.P$.
	 	$$\begin{array}{rclll}
	 		\smpv{(\sub{y}{x}{(\lb z.P)})} & \eq & \smpv{(\lb z.\sub{y}{x}{P})}\\
	 		& \eq & \lb z. \smp{(\sub{y}{x}{P})}\\
	 		& \eq & \lb z. \sub{\smpv{y}}{x}{\smp{P}} & \text{(by IH)}\\
	 		& \eq & \sub{\smpv{y}}{x}{\smpv{(\lb z.P)}}
	 	\end{array}$$
 	
 		Case $V=\CVW{W_1}{z}{W_2}$.
 		$$\begin{array}{rclll}
 			\smpv{(\sub{y}{x}{(\CVW{W_1}{z}{W_2})})} & \eq & \smpv{\CVW{\sub{y}{x}{W_1}}{z}{\sub{y}{x}{W_2}}}\\
 			& \eq & \sub{\smpv{(\sub{y}{x}{W_1})}}{z}{\smpv{(\sub{y}{x}{W_2})}}\\
 			& \eq & \sub{\sub{\smpv{y}}{x}{\smpv{W_1}}}{z}{(\sub{\smpv{y}}{x}{\smpv{W_2}})} & \text{(by IH)}\\
 			& \eq & \sub{\smpv{y}}{x}{(\sub{\smpv{W_1}}{z}{\smpv{W_2}})} & \text{(by Lemma \ref{lem:sub-of-sub})}\\
 			& \eq & \sub{\smpv{y}}{x}{\smpv{(\CVW{W_1}{z}{W_2}})}
 		\end{array}$$
 	
		The case $M = \CVN{W}{z}{N}$ is analogous to the previous one and the others follow by induction hypothesis just as case $V = \lb z.P$.

\end{proof}

\begin{thm}[Reflection in $\lbqmo$ of $\lbq$]\label{thm:reflection-of-simplified}\quad
\begin{enumerate}
	\item If $M\red N$ in $\lbq$ then $M\redd N$ in $\lbqmo$.
	\item If $M\red N$ in $\lbqmo$ then $\smp{M}\redd \smp{N}$ in $\lbq$.
	\item $M\redd \smp{M}$ in $\lbqmo$.
	\item For all $M\in\lbq$, $M\eq \smp{M}$.
\end{enumerate}
\end{thm}

\begin{proof} 
	
			The proof of the item 1 is by induction on the relation $M \red N$ in $\lbq$, where we must see $\Cut{M}{x}{N}$ as $\CMN{M}{x}{N}$. Cases $(\Bv)$, $(\etact)$, $(\pi_1)$ and $(\pi_2)$ are trivial. Case $(\sigmav)$ must be analyzed in two parts:		
			$$\begin{array}{rcll}
			\CMN{\rt{(\lb x.M)}}{y}{N} & \red_7 & \CVN{\lb x.M}{y}{N}\\
			& \redd & \sub{\lb x.M}{y}{N} & \text{(by  Lemma \ref{lem:subs})}
			\end{array}$$
			and
			$$\begin{array}{rcll}
				\CMN{\rt{x}}{y}{N} & \red_2 & \sub{x}{y}{N}
			\end{array}$$

			The proof of the item 2 is by induction on the relation $M \red N \in \lbqmo$. We proceed by cases:
			
			Case $(1)$.
				$$\begin{array}{rcll}
				\smp{\CMN{\rt{(\lb x.M)}}{y}{\li{y}{V}{z}{N}}} & \eq & \Cut{\rt{(\lb x. \smp{M})}}{y}{\li{y}{\smpv{V}}{z}{\smp{N}}}\\
				& \red_{\Bv} & \Cut{\Cut{\rt{\smpv{V}}}{x}{\smp{M}}}{z}{\smp{N}} \\ 
				& \eq & \smp{\Cut{\Cut{\rt{V}}{x}{M}}{z}{N}}
				\end{array}$$

			    Case $(2)$.
			    $$\begin{array}{rcll}
			    	\smp{\CMN{\rt{x}}{y}{N}} & \eq & \Cut{\rt{\smpv{x}}}{y}{\smp{N}}\\
			    	& \red_{\sigmav } & \sub{\smpv{x}}{y}{\smp{N}}\\
			    	& \eq & \smp{(\sub{x}{y}{N})} & \text{(by Lemma \ref{lem:maps-of-subs})}
			    \end{array}$$
			    
			    Case $(3)$.
			    $$\begin{array}{rcll}
			    	\smp{\CMN{M}{x}{\rt{x}}} & \eq & \Cut{\smp{M}}{x}{\rt{x}} \\
					& \red_{\etact} & \smp{M}
			    \end{array}$$

			    Case $(4)$.
			    $$\begin{array}{rcll}
			    	\smp{\CMN{\li{z}{V}{y}{P}}{x}{N}} & \eq &  \Cut{\li{z}{\smpv{V}}{y}{\smp{P}}}{x}{\smp{N}} \\
		    	    & \red_{\pi_1} & \li{z}{\smpv{V}}{y}{\Cut{\smp{P}}{x}{\smp{N}}}\\
		    	    & \eq  & \smp{\li{z}{V}{y}{\Cut{P}{x}{N}}}
			    \end{array}$$
			    
			     Case $(6)$ without side condition. 
			    $$\begin{array}{rcll}
			    	\smp{\CMN{\CMN{M}{y}{P}}{x}{N}} & \eq & \Cut{\Cut{\smp{M}}{y}{\smp{P}}}{x}{\smp{N}}\\
			    	& \red_{\pi_2} & \Cut{\smp{M}}{y}{\Cut{\smp{P}}{x}{\smp{N}}}\\
			    	& \eq & \smp{\CMN{M}{y}{\CMN{P}{x}{N}}}
			    \end{array}$$

			    Case $(7)$.
			    $$\begin{array}{rcll}
			       \smp{\CMN{\rt{(\lb x.M)}}{y}{N}} & \eq & \Cut{\rt{\smpv{(\lb x.M)}}}{y}{\smp{N}}\\
			    	& \red_{\sigmav} & \sub{\smpv{\lb x.M}}{y}{\smp{N}} & \text{($\smp{N}$ is not an $y$-covalue)}\\
			    	& \eq & \smp{\CVN{\lb x.M}{y}{N}}
			    \end{array}$$
		    
		    	 Case $(12)$.
		    	$$\begin{array}{rcll}
		    	 \smp{\CVN{V}{x}{\li{x}{W}{z}{N}}} & \eq & \sub{\smpv{V}}{x}{\smp{(\li{x}{W}{z}{N}})}\\
		    		& \eq & \Cut{\rt{\smpv{V}}}{x}{\li{x}{\sub{\smpv{V}}{x}{\smpv{W}}}{z}{\sub{\smpv{V}}{x}{\smp{N}}}}\\
		    		& \eq & \Cut{\smp{(\rt{V})}}{x}{\li{x}{\smpv{\CVW{V}{x}{W}}}{z}{\smp{\CVN{V}{x}{N}}}}\\
		    		& \eq & \smp{\CMN{\rt{V}}{x}{\li{x}{\CVW{V}{x}{W}}{z}{\CMN{V}{x}{N}}}}
		    	\end{array}$$
	    	
	    	The remaining cases are analogous to last one.

	    	Item 3 follows by simultaneous induction on $M$ and  $V \in \lbqmo$. We proceed by cases:
	    	
	    		Case $V=x$.
	    		$$\begin{array}{rclll}
	    			x & \redd & x &\eq & \smpv{x}
	    		\end{array}$$
    		
    			Case $V=\lb x.P$.
    			$$\begin{array}{rcll}
    				\lb x.P & \redd & \lb x.\smp{P} & \text{(by IH)}\\
    				& \eq & \smpv{(\lb x.P)}
    			\end{array}$$
    		
    			Case $V=\CVW{W_1}{x}{W_2}$.
    			$$\begin{array}{rcll}
    			  \CVW{W_1}{x}{W_2} & \redd & \sub{W_1}{x}{W_2} & \text{(by Lemma \ref{lem:subs})}\\
    				& \redd & \sub{\smpv{W_1}}{x}{\smpv{W_2}} & \text{(by IH)}\\
    				& \eq & \smpv{\CVW{W_1}{x}{W_2}}
    			\end{array}$$
    		
    			Case $M=\rt{W}$.
    			$$\begin{array}{rcll}
    				\rt{W} & \redd & \rt{\smpv{W}} &\text{(by IH)}\\
    				& \eq & \smp{(\rt{W})}
    			\end{array}$$
    		
    			Case $M=\li{x}{W}{y}{N}$.
    			$$\begin{array}{rcll}
    				\li{x}{W}{y}{N}  & \redd & \li{x}{\smpv{W}}{y}{\smp{N}} & \text{(by IH)}\\
    				& \eq & \smp{\li{x}{W}{y}{N}}
    			\end{array}$$
    		
    			Case $M=\CVN{W}{y}{N}$.
    			$$\begin{array}{rcll}
    				\CVN{W}{y}{N} & \redd & \sub{W}{y}{N} & \text{(by Lemma \ref{lem:subs})}\\
    				& \redd & \sub{\smpv{W}}{x}{\smp{N}} & \text{(by IH)}\\
    				& \eq & \smp{\CVN{W}{y}{N}}
    			\end{array}$$
    		
    			Case $M=\CMN{P}{y}{N}$.
    			$$\begin{array}{rcll}
    				\CMN{P}{y}{N} & \redd & \CMN{\smp{P}}{y}{\smp{N}} & \text{(by IH)}\\
    				& \eq & \Cut{\smp{P}}{y}{\smp{N}} & (\cutsymbol = \cutsymbol_3)\\
    				& \eq & \smp{\CMN{P}{y}{N}}
    			\end{array}$$

	    	Finally, item 4 is trivial if we consider $\cutsymbol = \cutsymbol_3$.
	    
\end{proof}

\begin{cor}[Conservativeness]
	For all $M,N\in\lbq$, $M\redd N$ in $\lbq$ iff $M\redd N$ in $\lbqmo$.
\end{cor}

\begin{proof}
	The left to the right direction follows immediately from the item 1 of the Theorem \ref{thm:reflection-of-simplified}.  Now suppose $M\redd N$ in $\lbqmo$, where $M,N\in\lbq$. 
	Then, $\smp{M}\redd \smp{N}$ in $\lbq$, by item 2. Finally, because $M\eq \smp{M}$ and $N\eq \smp{N}$, we get $M\redd N$ in $\lbq$.
\end{proof}

\begin{cor}[Equational correspondence between $\lbc$ and $\lbq$] (1) If $M\red N$ in $\lbc$ then $M^{\sharp\smp{}} \redd N^{\sharp\smp{}}$ in $\lbq$. (2) If $M\red N$ in $\lbq$ then $M^{\flat}\redd N^{\flat}$ in $\lbc$. (3) $M\redd M^{\flat\sharp\smp{}}$ in $\lbq$. (4) $M\conv M^{\sharp\smp{}\flat}$ in $\lbc$.
\end{cor}

\begin{proof}
	By composing Lemma $\ref{lem:ec-lc-lbmo}$ and Theorem $\ref{thm:reflection-of-simplified}$.
\end{proof}


Theorem \ref{thm:reflection-of-simplified} also gives a pre-Galois connection from $\lbq$ to $\lbqmo$. Thus, confluece of $\lbq$ can be pulled back from the confluence of $\lbqmo$. 
\begin{cor}
	$\lbq$ is confluent.
\end{cor}

The definition of the kernel $\LNF$ of $\lbq$ is collected in Table \ref{tab:kernel-lQ}.

\begin{table*}
		\centering
		
		$$\begin{array}{rrcl}
			\text{(terms)} & M, N &::= & \rt{V}  \, | \,  \li{x}{V}{y}{N}  \, | \,  \Cutv{V}{x}{N} \\
			\text{(values)}& V,W &::=& x \, | \, \lb x.M
		\end{array}$$
		
		$$\begin{array}{rrcll}
			(\Bv) & \Cutv{\lb x.M}{y}{\li{y}{V}{z}{N}} &  \to & \Cutv{V}{x}{\Cutvc{M}{z}{N}}
			& \text{if $y \notin FV(V) \cup FV(N)$} \\ 
			(\sigmav) & \Cutv{V}{y}{N} & \to & \sub{V}{y}{N} & \text{if $\Bv$ does not apply}
		\end{array}$$
		%
		\caption{The kernel of the $\lbq$-calculus, named $\LNF$}
		\label{tab:kernel-lQ}
	\end{table*}

There is a map $\knl{(\_)}:\lbq\to\LNF$ defined as follows:
$$\begin{array}{rclcrcl}
	\knl{(\rt{V})} & \eq & \rt{\knlv{V}}&&	\knlv{x} & \eq & x\\
	\knl{\li{x}{V}{y}{N}}  & \eq & \li{x}{\knlv{V}}{y}{\knl{N}}&&	\knlv{(\lb x.M)} & \eq & \lb x. \knl{M}\\
	\knl{\Cut{M}{y}{N}} & \eq & \Cutvc{\knl{M}}{y}{\knl{N}}
\end{array}$$


The next 3 Lemmas are needed to obtain Theorem \ref{thm:reflection-of-kernel}.

\begin{lem}\label{lem:perm-LNF}
	For all $M,N \in \LNF$, $\Cut{M}{x}{N} \redd_{\pi} \Cutc{M}{x}{N}$ in $\lbq$.
\end{lem}

\begin{proof}
	By induction on $M$.

		Case $M=\rt{V}$.
		$$\begin{array}{rcll}
			\Cut{\rt{V}}{x}{N} & \eq & \Cutc{\rt{V}}{x}{N}
		\end{array}$$
		
	Case $M=\li{y}{V}{z}{P}$.
		$$\begin{array}{rcll}
			 \Cut{\li{y}{V}{z}{P}}{x}{N} & \red_{\pi_1} &  \li{y}{V}{z}{\Cut{P}{x}{N}} \\
			& \redd & \li{y}{V}{z}{\Cutc{P}{x}{N}} & \text{(by IH)}\\
			& \eq & \Cutc{\li{y}{V}{z}{P}}{x}{N}
		\end{array}$$
		
		Case $M=\Cutv{V}{z}{P}$.
		$$\begin{array}{rcll}
			\Cut{\Cut{\rt{V}}{z}{P}}{x}{N} & \red_{\pi_2} & \Cut{\rt{V}}{z}{\Cut{P}{x}{N}} \\
			& \redd & \Cut{\rt{V}}{z}{\Cutc{P}{x}{N}} & \text{(by IH)}\\
			& \eq & \Cutc{\Cut{\rt{V}}{z}{P}}{x}{N}
		\end{array}$$

\end{proof}

\begin{cor}
For all $M\in\lbq$, $M\redd_{\pi}\knl{M}$.
\end{cor}
\begin{proof}
	By induction on $M$. The only interesting cases is $M=\Cut{M'}{x}{N}$, which follows by IH and Lemma \ref{lem:perm-LNF}.
\end{proof}

\begin{lem}\label{lem:subs-LQ-LNF}
	For all $V,W,M \in \lbq$:
	\begin{enumerate}
		\item $\knlv{(\sub{V}{x}{W})} \eq \sub{\knlv{V}}{x}{\knlv{W}}$.
		\item $\knl{(\sub{V}{x}{M})} \eq \sub{\knlv{V}}{x}{\knl{M}}$.
	\end{enumerate} 
\end{lem}

\begin{proof}
	By simultaneous induction on $W$ and $M$.
	
		Case $W=x$.
		$$\begin{array}{rcll}
			\knlv{(\sub{V}{x}{x})} & \eq & \knlv{V} \\
			& \eq & \sub{\knlv{V}}{x}{x}\\
			& \eq & \sub{\knlv{V}}{x}{\knlv{x}}
		\end{array}$$
	
		Case $W=y$.
		$$\begin{array}{rcll}
			\knlv{(\sub{V}{x}{y})} & \eq & \knlv{y} \\
			& \eq & y \\
			& \eq & \sub{\knlv{V}}{x}{y}\\
			& \eq & \sub{\knlv{V}}{x}{\knlv{y}}
		\end{array}$$
	
		Case $W=\lb y.P$.
		$$\begin{array}{rcll}
			\knlv{(\sub{V}{x}{(\lb y.P)})} & \eq & \knlv{(\lb y. \sub{V}{x}{P})} \\
			& \eq & \lb y. \knl{(\sub{V}{x}{P})}\\
			& \eq & \lb y. \sub{\knlv{V}}{x}{\knl{P}} & \text{(by IH)}\\
			& \eq & \sub{\knlv{V}}{x}{\knlv{(\lb y.P)}}
		\end{array}$$
	
		Case $M=\rt{W}$.
		$$\begin{array}{rcll}
			\knl{(\sub{V}{x}{(\rt{W})})} & \eq & \knl{(\rt{(\sub{V}{x}{W})})} \\
			& \eq & \rt{\knlv{(\sub{V}{x}{W})}}\\
			& \eq &\rt{( \sub{\knlv{V}}{x}{\knlv{W}})} & \text{(by IH)}\\
			& \eq & \sub{\knlv{V}}{x}{\knl{(\rt{W})}}
		\end{array}$$
	
		The remaining cases are analogous.

\end{proof}

\begin{lem}\label{lem:id-LNF}
	For all $M \in \LNF$, $\Cutvc{M}{y}{y} \redd M$ in $\LNF$.
\end{lem}

\begin{proof}
	By induction on $M$.

		Case $M=\rt{V}$.
		$$\begin{array}{rcll}
			\Cutvc{\rt{V}}{y}{y} & \eq & \Cutv{V}{y}{y} \\
			& \red_{\sigmav} & \sub{V}{y}{y}\\
			& \eq & V
		\end{array}$$
	
	Case $M=\li{x}{V}{z}{N}$.
		$$\begin{array}{rcll}
			\Cutvc{\li{x}{V}{z}{N}}{y}{y} & \eq & \li{x}{V}{z}{\Cutvc{N}{y}{y}} \\
			& \redd & \li{x}{V}{z}{\Cutv{N}{y}{y}} & \text{(by IH)}
		\end{array}$$
	
	Case $M=\Cutv{W}{z}{N}$.
		$$\begin{array}{rcll}
			\Cutvc{\Cutv{W}{z}{N}}{y}{y} & \eq & \Cutv{W}{z}{\Cutvc{N}{y}{y}} \\
			& \redd & \Cutv{W}{z}{\Cutv{N}{y}{y}} & \text{(by IH)}
		\end{array}$$

\end{proof}

\begin{thm}[Reflection in $\lbq$ of $\LNF$]\label{thm:reflection-of-kernel}\quad
	\begin{enumerate}
		\item If $M\red N$ in $\LNF$ then $M\redd N$ in $\lbq$.
		\item If $M\red N$ in $\lbq$ then $\knl{M}\redd \knl{N}$ in $\LNF$.
		\item $M\redd \knl{M}$ in $\lbq$.
		\item For all $M\in\LNF$, $M\eq \knl{M}$
	\end{enumerate}
\end{thm}

\begin{proof}
	The proof of the item 1 is by induction on the relation $M\red N$ in $\LNF$, where we must think of $\Cutv{V}{x}{M}$ as $\Cut{\rt{V}}{x}{M}$. We proceed by cases:

		Case $(\Bv)$.
		$$\begin{array}{rcll}
		  \Cut{\rt{(\lb x.M)}}{y}{\li{y}{V}{z}{N}} & \red_{\Bv} & \Cut{\Cut{\rt{V}}{x}{M}}{z}{N} \\
			& \redd & \Cutc{\Cut{\rt{V}}{x}{M}}{z}{N} & \text{(by Lemma \ref{lem:perm-LNF})}\\
			& \eq & \Cut{\rt{V}}{x}{\Cutc{M}{z}{N}}
		\end{array}$$
	
	Case $(\sigmav)$.
		$$\begin{array}{rcll}
			\Cut{\rt{(\lb x.M)}}{y}{\rt{y}} & \red_{\id} & \lb x.M \\
			& \eq & \sub{\lb x.M}{y}{y}
		\end{array}$$
		and 
		$$\begin{array}{rcll}
			\Cut{\rt{V}}{y}{N} & \red_{\sigmav} & \sub{V}{y}{N}
		\end{array}$$

	The proof of the item 2 is by induction on the relation $M\red N$ in $\lbq$. We proceed by cases:
	
Case $(\Bv)$.
		$$\begin{array}{rcll}
			 \knl{\Cut{\rt{(\lb x.M)}}{y}{\li{y}{V}{z}{N}}} & \eq & \Cutv{\lb x. \knl{M}}{y}{\li{y}{\knlv{V}}{z}{\knl{N}}}\\
			& \red_{\Bv} & \Cut{\knlv{V}}{x}{\Cutc{\knl{M}}{z}{\knl{N}}}\\
			& \eq &\Cutvc{\Cutv{\knlv{V}}{x}{\knl{M}}}{z}{\knl{N}}\\
			& \eq & \knl{\Cut{\Cut{\rt{V}}{x}{M}}{z}{N}}
		\end{array}$$
	
Case $(\sigmav)$.
		$$\begin{array}{rcll}
			\knl{\Cut{\rt{V}}{y}{N}} & \eq & \Cutvc{\knlv{V}}{y}{\knl{N}}\\
			& \red_{\sigmav} & \sub{\knlv{V}}{y}{\knl{N}}\\
			& \eq & \knl{(\sub{V}{y}{N})} & \text{(by Lemma \ref{lem:subs-LQ-LNF})}
		\end{array}$$
	
		Case $(\id)$.
		$$\begin{array}{rcll}
			\knl{\Cut{M}{y}{\rt{y}}} & \eq & \Cutvc{\knl{M}}{y}{\rt{y}}\\
			& \redd & \knl{M} & \text{(by Lemma \ref{lem:id-LNF})}
		\end{array}$$
	
	Case $(\pi_1)$.
		$$\begin{array}{rcll}
		  \knl{\Cut{\li{z}{V}{y}{P}}{x}{N}} & \eq & \Cutvc{\li{z}{\knlv{V}}{y}{\knl{P}}}{x}{\knl{N}} \\
			& \eq & \li{z}{\knlv{V}}{y}{\Cutvc{\knl{P}}{x}{\knl{N}}}\\
			& \eq & \knl{\li{z}{V}{y}{\Cut{P}{x}{N}}}
		\end{array}$$
	
	 Case $(\pi_2)$ is analogous to the previous one.

	Item 3 follow by simultaneous induction on $V$ and $M \in \lbq$. We proceed by cases:

	Case $V=x$.
		$$\begin{array}{rclll}
			x & \redd & x &\eq & \knlv{x}
		\end{array}$$
		
	Case $V=\lb x.P$.
		$$\begin{array}{rcll}
			\lb x.P & \redd & \lb x.\knl{P} & \text{(by IH)}\\
			& \eq & \knlv{(\lb x.P)}
		\end{array}$$
	
		Case $M=\rt{W}$.
		$$\begin{array}{rcll}
			\rt{W} & \redd & \rt{\knlv{W}} & \text{(by IH)}\\
			& \eq & \knl{(\rt{W})}
		\end{array}$$
		
		Case $M=\li{x}{W}{y}{N}$.
		$$\begin{array}{rcll}
			\li{x}{W}{y}{N}  & \redd & \li{x}{\knlv{W}}{y}{\knl{N}} & \text{(by IH)}\\
			& \eq & \knl{\li{x}{W}{y}{N}}
		\end{array}$$
	
	Case $M=\Cut{M}{y}{N}$.
		$$\begin{array}{rcll}
			\Cut{M}{y}{N}  & \redd & \Cut{\knl{M}}{y}{\knl{N}} & \text{(by IH)}\\
			& \redd & \Cutc{\knl{M}}{y}{\knl{N}} & \text{(by Lemma \ref{lem:perm-LNF})}\\
			& \eq & \knl{\Cut{M}{y}{N}}
		\end{array}$$

	If we consider $\Cutv{V}{x}{M} = \Cut{\rt{V}}{x}{M}$, then item 4 is trivial.
	
\end{proof}

\begin{cor}
	Reduction in $\LNF$ is confluent.
\end{cor}

\subsection{Proofs of Section \ref{sec:VFS}}

The next 3 Lemmas are needed to obtain Theorem \ref{thm:simul-vfs}.

\begin{lem}\label{lem:vfs-id}
	For all $M \in \VFS$, $\Cutvfsc{M}{z.\rt{z}} \redd M$ in $\VFS$.
\end{lem}

\begin{proof}
	By induction on $M$.
	
	Case $M=\rt{V}$.
		$$\begin{array}{rcll}
			\Cutvfsc{\rt{V}}{z.\rt{z}} & \eq & \Cutvfs{V}{z.\rt{z}}\\
			& \red_{\sigmav} & \sub{V}{z}{(\rt{z})}\\
			& \eq & \rt{V}
		\end{array}$$
	
	    Case $M=\Cutvfs{V}{x.N}$.
		$$\begin{array}{rcll}
			\Cutvfsc{\Cutvfs{V}{x.N}}{z.\rt{z}} & \eq & \Cutvfs{V}{(x.N):z.\rt{z}}\\
			& \eq & \Cutvfs{V}{x.\Cutvfsc{N}{z.\rt{z}}}\\
			& \redd &\Cutvfs{V}{x.N} & \text{(by IH)}
		\end{array}$$
	
		Case $M=\Cutvfs{V}{\garg{W}{x}{N}}$.
		$$\begin{array}{rcll}
			\Cutvfsc{\Cutvfs{V}{\garg{W}{x}{N}}}{z.\rt{z}} & \eq & \Cutvfs{V}{\garg{W}{x}{N}:z.\rt{z}}\\
			& \eq & \Cutvfs{V}{\garg{W}{x}{\Cutvfsc{N}{z.\rt{z}}}}\\
			& \redd &  \Cutvfs{V}{\garg{W}{x}{N}} & \text{(by IH)}
		\end{array}$$
	
\end{proof}

\begin{lem}\label{lem:vfs-red}
	For all $M \in \lbc$, $N,N' \in \VFS$, if $N \red N'$ in $\VFS$ then $\vfsc{M}{x.N} \red \vfsc{M}{x.N'}$ in $\VFS$.
\end{lem}

\begin{proof}
	By induction on $M$.
	
		Case $M=V$.
		$$\begin{array}{rcll}
			\vfsc{V}{x.N} & \eq & \Cutvfs{\vfsv{V}}{x.N}\\
			& \red & \Cutvfs{\vfsv{V}}{x.N'} & \text{(by hypothesis)}\\
			& \eq & \vfsc{V}{x.N'}.
		\end{array}$$
	
		Case $M=PQ$, where $P$ is not a value.
		$$\begin{array}{rcll}
			\vfsc{PQ}{x.N} &\eq& \vfsc{P}{m.\vfsc{Q}{n.\Cutvfs{\vfsv{m}}{\garg{\vfsv{n}}{x}{N}}}}
		\end{array}$$
		$N_1 \eq\Cutvfs{\vfsv{m}}{\garg{\vfsv{n}}{x}{N}} \red \Cutvfs{\vfsv{m}}{\garg{\vfsv{n}}{x}{N'}} \eq N_1'$, by hypothesis. Because of that, $N2 \eq \vfsc{Q}{n.N_1} \red \vfsc{Q}{n.N_1'} \eq N_2'$, by induction hypothesis. Consequently, $\vfsc{P}{m.N_2} \red \vfsc{P}{m.N_2'}$, once more by induction hypothesis, and $\vfsc{P}{m.N_2'} \eq \vfsc{PQ}{x.N'}.\\$
		
		The remaining cases are analogous.
	
\end{proof}

\begin{lem}\label{lem:vfs-subs}
	For all $V,W, M \in \lbc$, $P \in \VFS$:
	\begin{enumerate}
		\item $\vfsv{(\sub{V}{x}{W})} \eq \sub{\vfsv{V}}{x}{\vfsv{W}}$.
		\item $\sub{\vfsv{V}}{x}{\vfsc{M}{y.P}} \eq \vfsc{\sub{V}{x}{M}}{y.\sub{\vfsv{V}}{x}{P}}$.
		\item $\vfst{(\sub{V}{x}{M})} \eq \sub{\vfsv{V}}{x}{\vfst{M}}$
	\end{enumerate}
\end{lem}

\begin{proof}
	Notice that, if item 2 holds for some $M$, then item 3 holds for the same $M$. Item 3 is obtained using item 2.
	$$\begin{array}{rcll}
		\vfst{(\sub{V}{x}{M})} & \eq & \vfsc{\sub{V}{x}{M}}{z.\rt{z}}\\
		& \eq & \vfsc{\sub{V}{x}{M}}{z.\sub{\vfsv{V}}{x}{(\rt{z})}}\\
		& \eq & \sub{\vfsv{V}}{x}{\vfsc{M}{z.\rt{z}}} & \text{(by IH)}\\
		& \eq & \sub{\vfsv{V}}{x}{\vfst{M}}
	\end{array}$$

	Hence, it suffices to prove items 1 and 2, and this is done by simultaneous induction on $W$ and $M$.

	Case $W=x$.
		$$\begin{array}{rcll}
			\vfsv{(\sub{V}{x}{x})} & \eq & \vfsv{V}\\
			& \eq & \sub{\vfsv{V}}{x}{x}\\
			& \eq & \sub{\vfsv{V}}{x}{\vfsv{x}}
		\end{array}$$
	
	Case $W=y$.
		$$\begin{array}{rcll}
			\vfsv{(\sub{V}{x}{y})} & \eq & \vfsv{y}\\
			& \eq & y\\
			& \eq & \sub{\vfsv{V}}{x}{y}\\
			& \eq & \sub{\vfsv{V}}{x}{\vfsv{y}}
		\end{array}$$
	
	Case $W=\lb y.P$.
		$$\begin{array}{rcll}
			\vfsv{(\sub{V}{x}{(\lb y.P)})} & \eq & \vfsv{(\lb y.\sub{V}{x}{P})}\\
			& \eq & \lb y. \vfst{(\sub{V}{x}{P})}\\
			& \eq & \lb y. \sub{\vfsv{V}}{x}{\vfst{P}} & \text{(by IH)}\\
			& \eq & \sub{\vfsv{V}}{x}{\vfsv{(\lb y.P)}}
		\end{array}$$
	
		Case $M=W$.
		$$\begin{array}{rcll}
			 \sub{\vfsv{V}}{x}{\vfsc{W}{y.P}}& \eq & \sub{\vfsv{V}}{x}{\Cutvfs{\vfsv{W}}{y.P}}\\
			& \eq & \Cutvfs{\sub{\vfsv{V}}{x}{\vfsv{W}}}{y.\sub{\vfsv{V}}{x}{P}}\\
			& \eq & \Cutvfs{\vfsv{(\sub{V}{x}{W})}}{y.\sub{\vfsv{V}}{x}{P}} & \text{(by IH)}\\
			& \eq & \vfsc{\sub{V}{x}{W}}{y.\sub{\vfsv{V}}{x}{P}}
		\end{array}$$
	
	Case $M=WW'$.
		$$\begin{array}{rcll}
			 \sub{\vfsv{V}}{x}{\vfsc{(WW')}{y.P}}	& \eq & \sub{\vfsv{V}}{x}{\Cutvfs{\vfsv{W}}{\garg{\vfsv{W'}}{y}{N}}}\\
			& \eq & \Cutvfs{\sub{\vfsv{V}}{x}{\vfsv{W}}}{\garg{\sub{\vfsv{V}}{x}{\vfsv{W'}}}{y}{\sub{\vfsv{V}}{x}{N}}}\\ 
			& \eq & \Cutvfs{\vfsv{(\sub{V}{x}{W})}}{\garg{\vfsv{(\sub{V}{x}{W'})}}{y}{\sub{\vfsv{V}}{x}{N}}} & \text{(by IH)}\\
			& \eq & \vfsc{\sub{V}{x}{W}\sub{V}{x}{W'}}{y. \sub{\vfsv{V}}{x}{N}}\\
		 	& \eq & \vfsc{\sub{V}{x}{(WW')}}{y. \sub{\vfsv{V}}{x}{N}}
		\end{array}$$
	
		Case $M=WQ$, where $Q$ is not a value.
		$$\begin{array}{rcll}
			 \sub{\vfsv{V}}{x}{\vfsc{(WQ)}{y.P}} & \eq & \sub{\vfsv{V}}{x}{\vfsc{Q}{n.\Cutvfs{\vfsv{W}}{\garg{\vfsv{n}}{y}{P}}}}\\
			& \eq & \vfsc{\sub{V}{x}{Q}}{n.\sub{\vfsv{V}}{x}{\Cutvfs{\vfsv{W}}{\garg{\vfsv{n}}{x}{P}}}} & \text{(by IH)}\\
			& \eq & \vfsc{\sub{V}{x}{Q}}{n.\Cutvfs{\sub{\vfsv{V}}{x}{\vfsv{W}}}{\garg{\vfsv{n}}{y}{\sub{\vfsv{V}}{x}{P}}}}\\
			& \eq & \vfsc{\sub{V}{x}{Q}}{\Cutvfs{\vfsv{(\sub{V}{x}{W})}}{\garg{\vfsv{n}}{y}{\sub{\vfsv{V}}{x}{P}}}} & \text{(by IH)}\\
			& \eq & \vfsc{\sub{V}{x}{W}\sub{V}{x}{Q}}{y.\sub{\vfsv{V}}{x}{P}}\\
			& \eq & \vfsc{\sub{V}{x}{(WQ)}}{y.\sub{\vfsv{V}}{x}{P}}
		\end{array}$$ 
		
		Cases $M=NQ$, where $N$ is not a value, and $M=\lt{z}{N}{Q}$ are analogous.
		
\end{proof}

\noindent \textbf{Theorem \ref{thm:simul-vfs}} (Simulation by VFS-translation)\textbf{.}
\begin{enumerate}
	\item Let $R\in\{B,\ltv,\etalt\}$. If $M\red_R N$ in $\lbc$ then $\vfst{M}\redd\vfst{N}$ in $\VFS$.
	\item Let $R\in\{\ltmn,\ltvn,\assoc\}$. If $M\red_R N$ in $\lbc$ then $\vfst{M}\eq\vfst{N}$ in $\VFS$.
\end{enumerate}

\begin{proof}
	Item 1. Let $R\in\{B,\ltv,\etalt\}$. The proof is by induction on the relation $M\red_R N$. The base cases are as follows.	

		Case $(B)$.
		$$\begin{array}{rcll}
			 \vfst{((\lb x.M)V)} & \eq &  \Cutvfs{\lb x. \vfst{M}}{\garg{\vfsv{V}}{z}{\rt{z}}}\\
			& \red_{\Bv} & \Cutvfs{\vfsv{V}}{x.\Cutvfsc{\vfst{M}}{z.\rt{z}}}\\
			& \redd & \Cutvfs{\vfsv{V}}{x.\vfst{M}} & \text{(by Lemma \ref{lem:vfs-id})}\\
			& \eq & \Cutvfs{\vfsv{V}}{x. \vfsc{M}{z.\rt{z}}}\\
			& \eq & \vfsc{V}{x.\vfsc{M}{z.\rt{z}}}\\
			& \eq & \vfst{(\lt{x}{V}{M})}
		\end{array}$$
	
		On the other hand, suppose $Q$ is not a value.
		$$\begin{array}{rcll}
			\vfst{((\lb x.M)Q)}& \eq & \vfsc{Q}{n.\Cutvfs{\lb x.\vfst{M}}{\garg{n}{z}{\rt{z}}}}\\
			& \red_{\Bv}& \vfsc{Q}{n.\Cutvfs{n}{x.(\vfst{M}:z.\rt{z})}}\\ 
			& \eq & \vfsc{Q}{n.\Cutvfs{n}{(x.\vfst{M}:z.\rt{z})}}\\
			& \eq &\vfsc{Q}{n.\Cutvfsc{\Cutvfs{n}{x.\vfst{M}}}{z.\rt{z}}} & \text{(by Lemma \ref{lem:vfs-red})}\\
			& \redd & \vfsc{Q}{n.\Cutvfs{n}{x.\vfst{M}}} & \text{(by Lemma \ref{lem:vfs-id})}\\
			& \red_{\sigmav} & \vfsc{Q}{n.\sub{n}{x}{\vfst{M}}} & \text{(by Lemma \ref{lem:vfs-red})}\\
			& \eq_\alpha & \vfsc{Q}{x.\vfst{M}}
 		\end{array}$$
 	
 		Case $(\ltv)$.
 		$$\begin{array}{rcll}
 			 \vfst{(\lt{x}{V}{M})}& \eq & \vfsc{V}{x.\vfsc{M}{z.\rt{z}}}\\
 			& \eq & \Cutvfs{\vfsv{V}}{x.\vfsc{M}{z.\rt{z}}}\\
 			& \eq & \Cutvfs{\vfsv{V}}{x.\vfst{M}}\\
 			& \red_{\sigmav} & \sub{\vfsv{V}}{x}{\vfst{M}}\\
 			& \eq & \vfst{(\sub{V}{x}{M})} & \text{(by Lemma \ref{lem:vfs-subs})}
 		\end{array}$$
 	
 		Case $(\etalt)$.
 		$$\begin{array}{rcll}
 			 \vfst{(\lt{x}{M}{x})}& \eq & \vfsc{M}{x.\vfsc{x}{z.\rt{z}}}\\
 			& \eq & \vfsc{M}{x.\Cutvfs{x}{z.\rt{z}}}\\
 			& \red_{\sigmav} & \vfsc{M}{x.\sub{x}{z}{(\rt{z})}} & \text{(by Lemma \ref{lem:vfs-red})}\\
 			& \eq & \vfsc{M}{x.\rt{x}}\\
 			& \eq & \vfst{M}
 		\end{array}$$

    The inductive cases are easy, because $\redd$ is compatible.
    
	Item 2. Let $R\in\{\ltmn,\ltvn,\assoc\}$. The proof is by induction on the relation $M\red_R N$. The base cases are as follows:

		Case $(\ltmn)$.
		$$\begin{array}{rcll}
			\vfst{(MN)} & \eq & \vfsc{M}{m.\vfsc{mN}{z.\rt{z}}}\\
			& \eq_\alpha &   \vfsc{M}{x.\vfsc{xN}{z.\rt{z}}}\\
			& \eq & \vfst{(\lt{x}{M}{xN})}
		\end{array}$$
	
		Case $(\ltvn)$.
		$$\begin{array}{rcll}
			\vfst{(VN)} & \eq & \vfsc{N}{n.\vfsc{Vn}{z.\rt{z}}}\\
			& \eq_\alpha &   \vfsc{N}{x.\vfsc{Vx}{z.\rt{z}}}\\
			& \eq & \vfst{(\lt{x}{N}{Vx})}
		\end{array}$$
	
	    Case $(\assoc)$.
		$$\begin{array}{rcll}
			 \vfst{(\lt{y}{\lt{x}{M}{N}}{P})} & \eq & \vfsc{\lt{x}{M}{N}}{y.\vfsc{P}{z.\rt{z}}}\\
			& \eq & \vfsc{M}{x.\vfsc{N}{y.\vfsc{P}{z.\rt{z}}}}\\
			& \eq & \vfsc{M}{x.\vfsc{\lt{y}{N}{P}}{z.\rt{z}}}\\
			& \eq & \vfst{(\lt{x}{M}{\lt{y}{N}{P}})}
		\end{array}$$
	    
	    The inductive cases are trivial.

\end{proof}

\noindent \textbf{Theorem \ref{thm:decomposition}} (Decomposition of the CPS-translation)\textbf{.}
	\begin{enumerate}
		\item For all $V\in\lbc$, $\ngv{{\vfsv{V}}}=\cpsv{V}$.
		\item For all $M\in\lbc$, $N\in\VFS$, $\angt{\vfsc{M}{x.N}}=\cpsc{M}{\lb x.\angt{N}}$.
		\item For all $M\in\lbc$, $\angt{{\vfst{M}}}=\cpsk{M}$.
		\item For all $M\in\lbc$, $\ngt{{\vfst{M}}}=\cpst{M}$.
	\end{enumerate}	

\begin{proof}
	Notice that, if item 2 holds for some $M$, then item 3 and 4 hold for the same $M$. Item 3 is obtained using item 2.
	$$
	\begin{array}{rcll}
		\angt{{\vfst{M}}}&=&\angt{\vfsc M{x.\rt{x}}}\\
		&=&\cpsc{M}{\lb x.\angt{(\rt{x})}}&\text{(by item 2)}\\
		&=&\cpsc{M}{\lb x.kx}\\
		&=&\cpsk{M}
	\end{array}
	$$
	Item 4 is obtained using item 3.
	$$
	\begin{array}{rcll}
		\ngt{{\vfst{M}}}&=&\lb k.\angt{{\vfst{M}}}\\
		&=&\lb k.\cpsk{M}&\text{(by item 3)}\\
		&=&\cpst{M}
	\end{array}
	$$
	Hence, it suffices to prove items 1 and 2, and this is done by simultaneous induction on $V$ and $M$.
	
	Case $V=x$.
	$$
	\begin{array}{rcll}
		\ngv{{\vfsv{x}}}&=&\ngv x\\
		&=&x\\
		&=&\cpsv x
	\end{array}
	$$
	
	Case $V=\lb x.M$. By IH, item 4 holds of $M$. Then:
	$$
	\begin{array}{rcll}
		\ngv{{\vfsv{(\lb x.M)}}}&=&\ngv{(\lb x.\vfst{M})}\\
		&=&\lb x.\ngt{{\vfst{M}}}\\
		&=&\lb x.\cpst{M}&\text{(by item 4)}\\
		&=&\cpsv{(\lb x.M)}
	\end{array}
	$$
	
	Case $M=V$.
	$$
	\begin{array}{rcll}
		\angt{\vfsc{V}{x.N}}&=&\angt{\Cutvfs{\vfsv V}{x.N}}\\
		&=&(\lb x.\angt M)\ngv{{\vfsv V}}\\
		&=&(\lb x.\angt M)\cpsv V & \text{(by IH)}\\
		&=&\cpsc{M}{\lb x.\angt{N}}
	\end{array}
	$$
	
	Case $M=PQ$, with $P$ not a value.
	$$
	\begin{array}{rcll}
		\angt{\vfsc{PQ}{x.N}}&=&\angt{\vfsc{P}{m.\vfsc{mQ}{x.N}}} \\
		&=& \cpsc{P}{\lb m.\angt{\vfsc{mQ}{x.N}}} & \text{(by IH)}\\
		&=& \cpsc{P}{\lb m.\cpsc{mQ}{\lb x.\angt{N}}} & \text{(by IH)}\\
		&=& \cpsc{PQ}{\lb x.\angt{N}}
	\end{array}
	$$
	
	Case $M=VQ$, with $Q$ not a value.
	$$
	\begin{array}{rcll}
		\angt{\vfsc{VQ}{x.N}}&=&\angt{\vfsc{Q}{n.\vfsc{Vn}{x.N}}} \\
		&=& \cpsc{Q}{\lb n.\angt{\vfsc{Vn}{x.N}}} & \text{(by IH)}\\
		&=& \cpsc{Q}{\lb n.\cpsc{Vn}{\lb x.\angt{N}}} & \text{(by IH)}\\
		&=& \cpsc{VQ}{\lb x.\angt{N}}
	\end{array}
	$$
	
	Case $M=VW$.
	$$
	\begin{array}{rcll}
		\angt{\vfsc{VW}{x.N}}&=& \angt{\Cutvfs{\vfsv{V}}{\garg{\vfsv{W}}{x}{N}}}\\
		&=& \ngv{{\vfsv{V}}}\ngv{{\vfsv{W}}}(\lb x.\angt N)\\
		&=& \cpsv{V}\cpsv{V}(\lb x.\angt N) & \text{(by IH)}\\
		&=& \cpsc{VW}{\lb x.\angt{N}}
	\end{array}
	$$
	
	Case $M=\lt{y}{P}{Q}$.
	$$
	\begin{array}{rcll}
		&&\angt{\vfsc{\lt{y}{P}{Q}}{x.N}}\\
		&=& \angt{\vfsc{P}{y.\vfsc{Q}{x.N}}}\\
		&=& \cpsc{P}{\lb y.\angt{\vfsc{Q}{x.N}}} & \text{(by IH)}\\
		&=& \cpsc{P}{\lb y.\cpsc{Q}{\lb x.\angt N}} & \text{(by IH)}\\
		&=& \cpsc{\lt{y}{P}{Q}}{\lb x.\angt{N}}
	\end{array}
	$$
	
\end{proof}


The next 4 Lemmas are needed to obtain Theorem \ref{thm:iso}.

\begin{lem}\label{lem:cong-subs}
	For all $V, W, M \in \VFS$:
	\begin{enumerate}
			\item $\ngv{(\sub{V}{x}{W})} = \sub{\ngv{V}}{x}{\ngv{W}}$.
			\item $\angt{(\sub{V}{x}{M})} = \sub{\ngv{V}}{x}{\angt{M}}$.
	\end{enumerate} 
\end{lem}

\begin{proof}
	By simultaneous induction on $W$ and $M$.
	
	Case $W=x$.
	$$
	\begin{array}{rcll}
		\ngv{(\sub{V}{x}{x})} & = & \ngv{V}\\
		& = & \sub{\ngv{V}}{x}{x}\\
		& = & \sub{\ngv{V}}{x}{\ngv{x}}
	\end{array}
	$$
	
	Case $W=y$.
	$$
	\begin{array}{rcll}
		\ngv{(\sub{V}{x}{y})} & = & \ngv{y}\\
		& = & y\\ 
		& = & \sub{\ngv{V}}{x}{y}\\
		& = & \sub{\ngv{V}}{x}{\ngv{y}}
	\end{array}
	$$
	
	Case $W=\lb y.P$.
	$$
	\begin{array}{rcll}
		\ngv{(\sub{V}{x}{(\lb y.P)})} & = & \lb y. \angt{(\sub{V}{x}{P})} \\
		& = & \lb y. \sub{\ngv{V}}{x}{\angt{P}} & \text{(by IH)}\\ 
		& = & \sub{\ngv{V}}{x}{\ngv{(\lb y.P)}}
	\end{array}
	$$
	
	Case $M=\rt{W}$.
	$$
	\begin{array}{rcll}
		\angt{(\sub{V}{x}{(\rt{W})})} & = & \angt{(\rt{(\sub{V}{x}{W})})}\\
		& = & k\ngv{(\sub{V}{x}{W})}\\
		& = & k\sub{\ngv{V}}{x}{\ngv{W}} & \text{(by IH)}\\
		& = & \sub{\ngv{V}}{x}{(k\ngv{W})}\\
		& = & \sub{\ngv{V}}{x}{\angt{(\rt{W})}}
	\end{array}
	$$
	
	The remaining cases are analogous.
	
\end{proof}


\begin{lem}\label{lem:cong-col}
	For all $M,N \in \VFS$, $\angt{\Cutvfsc{M}{y.N}} = \sub{\lb y. \angt{N}}{k}{\angt{M}}$.
\end{lem}

\begin{proof}
	By induction on $M$.
	
	Case $M=\rt{V}$.
	$$
	\begin{array}{rcll}
		\angt{\Cutvfsc{\rt{V}}{y.N}}& = & \angt{\Cutvfs{V}{y.N}} \\
		& = & (\lb y. \angt{N})\ngv{V}\\
		& = & \sub{\lb y. \angt{N}}{k}{k\ngv{V}}\\
		& = & \sub{\lb y. \angt{N}}{k}{\angt{(\rt{V})}}
	\end{array}
	$$
	
	Case $M=\Cutvfs{V}{x.M'}$.
	$$
	\begin{array}{rcll}
		\angt{\Cutvfsc{\Cutvfs{V}{x.M_0}}{y.N}}& = &\angt{\Cutvfs{V}{x.\Cutvfsc{M_0}{y.N}}}\\
		& = & (\lb x. \angt{\Cutvfsc{M_0}{y.N}})\ngv{V}\\
		& = & (\lb x. \sub{\lb y. \angt{N}}{k}{\angt{M_0}})\ngv{V} & \text{(by IH)}\\
		& = & \sub{\lb y. \angt{N}}{k}{((\lb x. \angt{M_0})\ngv{V})} & (k \notin \ngv{V})\\
		& = & \sub{\lb y. \angt{N}}{k}\angt{\Cutvfs{V}{x.M'}} 
    \end{array}
	$$
	
	Case $M=\Cutvfs{V}{\garg{W}{x}{M_0}}$.
	$$
	\begin{array}{rcll}
		 \angt{(\Cutvfs{V}{\garg{W}{x}{M_0}}:y.N)} & = & \angt{\Cutvfs{V}{\garg{W}{x}{\Cutvfsc{M_0}{y.N}}}} \\
		& = & \ngv{V}\ngv{W}(\lb x. \angt{\Cutvfsc{M_0}{y.N}}) \\
		& = &  \ngv{V}\ngv{W}(\lb x. \sub{\lb y. \angt{N}}{k}{\angt{M_0}}) & \text{(by IH)}\\
		& = & \sub{\lb y. \angt{N}}{k}{(\ngv{V}\ngv{W}(\lb x. \angt{M_0}))} & (k \notin \ngv{V}, \ngv{W}) \\
		& = & \sub{\lb y. \angt{N}}{k}\angt{\Cutvfs{V}{\garg{W}{x}{M_0}}} 
	\end{array}
	$$
	
\end{proof}


\begin{lem}\label{lem:cong-subs-CPS}
	For all $V,W,M,P \in \CPS$:
	\begin{enumerate}
		\item $\posv{(\sub{V}{x}{W})} = \sub{\posv{V}}{x}{\posv{W}}$.
		\item $\apos{(\sub{V}{x}{M})} = \sub{\posv{V}}{x}{\apos{M}}$.
		\item $\post{(\sub{V}{x}{P})} = \sub{\posv{V}}{x}{\post{P}}$.
	\end{enumerate} 
\end{lem}

\begin{proof}
By simultaneous induction on $W, M$ and $P$.
	
	Case $W=x$.
	$$
	\begin{array}{rcll}
		\posv{(\sub{V}{x}{x})} & = & \posv{V}\\
		& = & \sub{\posv{V}}{x}{x}\\
		& = & \sub{\posv{V}}{x}{\posv{x}}
	\end{array}
	$$
	
	Case $W=y$.
	$$
	\begin{array}{rcll}
		\posv{(\sub{V}{x}{y})} & = & \posv{y}\\
		& = & y\\
		& = & \sub{\posv{V}}{x}{y}\\
		& = & \sub{\posv{V}}{x}{\posv{y}}
	\end{array}
	$$
	
	Case $W=\lb y.P$.
	$$
	\begin{array}{rcll}
		\posv{(\sub{V}{x}{(\lb y.P)})} & = & \posv{(\lb y. \sub{V}{x}{P})}\\
		& = & \lb y. \post{(\sub{V}{x}{P})}\\
		& = & \lb y. \sub{\posv{V}}{x}{\post{P}} & \text{(by IH)}\\
		& = &  \sub{\posv{V}}{x}{(\lb y.\post{P})}\\
		& = & \sub{\posv{V}}{x}{\posv{(\lb y.P)}}
	\end{array}
	$$
	
	Case $M=kW$.
	$$
	\begin{array}{rcll}
		\apos{(\sub{V}{x}{(kW)})} & = & \apos{(k\sub{V}{x}{W})}\\
		& = & \rt{(\posv{\sub{V}{x}{W}})}\\
		& = & \rt{(\sub{\posv{V}}{x}{\posv{W}})} & \text{(by IH)}\\
		& = &  \sub{\posv{V}}{x}{(\rt{\posv{W}})}\\
		& = & \sub{\posv{V}}{x}{\apos{(kW)}}
	\end{array}
	$$
	
	Case $P=\lb k.M$.
	$$
	\begin{array}{rcll}
		\post{(\sub{V}{x}{(\lb k.P))})} & = & \post{(\lb k. \sub{V}{x}{M})}\\
		& = & \lb k. \apos{\sub{V}{x}{M}} \\
		& = &  \lb k. \sub{\posv{V}}{x}{\apos{M}} & \text{(by IH)}\\
		& = &  \sub{\posv{V}}{x}{(\lb k. \apos{M})}\\
		& = & \sub{\posv{V}}{x}{\post{(\lb k.M)}}
	\end{array}
	$$
	
	Cases $M=(\lb x.N)V$ and $M=W_1W_2(\lb x.N)$ are analogous to the case $M=kW$.
	
\end{proof}

\begin{lem}\label{lem:cong-col-CPS}
	For all $M,N \in \CPS$, $ \apos{(\sub{\lb y.N}{k}{M})} = \Cutvfsc{\apos{M}}{y.\apos{N}}$.
\end{lem}

\begin{proof}
	By induction on $M$.
	
	Case $M=kV$.
	$$
	\begin{array}{rcll}
		 \apos{(\sub{\lb y.N}{k}{(kV)})} & = & \apos{(\lb y.M)V}\\
		 & = & \Cutvfs{\posv{V}}{y.\apos{N}}\\
		 &=& \Cutvfsc{\rt{\posv{V}}}{y.\apos{N}}\\
		 &=& \Cutvfsc{\apos{(kV)}}{y.\apos{N}}
	\end{array}
	$$
	
	Case $M=(\lb x.M_0)V$.
	$$
	\begin{array}{rcll}
		 \apos{(\sub{\lb y.N}{k}{((\lb x.M_0)V)})}& = & \apos{((\lb x. \sub{\lb y.N}{k}{M_0})V)} & (k \notin V)\\
		& = & \Cutvfs{\posv{V}}{x.\apos{(\sub{\lb y.N}{k}{M_0})}}\\
		& = & \Cutvfs{\posv{V}}{x.\Cutvfsc{\apos{M_0}}{y.\apos{N}}} & \text{(by IH)}\\
	     &=& \Cutvfs{\posv{V}}{(x.\apos{M_0} : y. \apos{N})}\\
	     &=& \Cutvfsc{\Cutvfs{\posv{V}}{x.\apos{M_0}}}{y.\apos{N}}\\
	     &=& \Cutvfsc{\apos{((\lb x.M_0)V)}}{y.\apos{N}}
	\end{array}
	$$
	
	The remaining case is analogous.
	
\end{proof}

\noindent \textbf{Theorem \ref{thm:iso}} (\VFS $\cong$ \CPS )\textbf{.}
\begin{enumerate}
	\item For all $M,V\in\VFS$, $\post{{\ngt{M}}}=M$ and $\apos{{\angt{M}}}=M$ and $\posv{{\ngv{V}}}=V$.
	\item For all $P,M,V\in\CPS$, $\ngt{{\post{P}}}=P$ and $\angt{{\apos{M}}}=M$ and $\ngv{{\posv{V}}}=V$.
	\item If $M_1\to M_2$ in $\VFS$ then $\angt{M_1}\to\angt{M_2}$ in $\CPS$ (hence $\ngt{M_1}\to\ngt{M_2}$ in $\CPS$).
	\item If $M_1\to M_2$ in $\CPS$ then $\apos{M_1}\to\apos{M_2}$ in $\VFS$. Hence If $P_1\to P_2$ in $\CPS$ then $\post{P_1}\to\post{P_2}$ in $\VFS$.
\end{enumerate}

\begin{proof}
	Item 1 is straightforward by simultaneous induction on $M$ and $V$. Similarly, item 2 is straightforward by simultaneous induction on $P, M$ and $V$. Item 3 follows by induction on the relation $M_1 \to M_2$ in $\VFS$. We just show the base cases.
	
	Case $(\sigmav)$.
	$$
	\begin{array}{rcll}
		\angt{\Cutvfs{V}{x.M}} & = & (\lb x.\angt{M})\ngv{V}\\
		& \red_{\sigmav} & \sub{\ngv{V}}{x}{\angt{M}}\\
		& = & \angt{(\sub{V}{x}{M})} & \text{(by Lemma \ref{lem:cong-subs})}
	\end{array}
	$$
	
	Case $(\Bv)$.
	$$
	\begin{array}{rcll}
		 \angt{\Cutvfs{\lb x.M}{\garg{V}{y}{N}}} & = & (\lb x. \ngt{M})\ngv{V}(\lb y.\angt{N}) \\
		& \red_{\Bv} & (\lb x. \sub{\lb y.\angt{N}}{k}{\angt{M}})\ngv{V}\\
		& = & (\lb x. \angt{\Cutvfsc{M}{y.N}})\ngv{V} & \text{(by Lemma \ref{lem:cong-col})} \\
		& = & \angt{\Cutvfs{V}{x.\Cutvfsc{M}{y.N}}}
	\end{array}
	$$
	
	Item 4 follows by induction on the relation $M_1 \to M_2$ in $\CPS$. We just show the base cases.
	
	Case $(\sigmav)$.
	$$
	\begin{array}{rcll}
		\apos{(\lb x.M)V} & = & \Cutvfs{\posv{V}}{x.\apos{M}}\\
		& \red_{\sigmav} & \sub{\posv{V}}{x}{\apos{M}}\\
		& = & \apos{(\sub{V}{x}{M})} & \text{(by Lemma \ref{lem:cong-subs-CPS})}
	\end{array}
	$$
	
	Case $(\Bv)$, with $K = \lb y.N$.
	$$
	\begin{array}{rcll}
		\apos{((\lb xk.M)W(\lb y.N))} & = &  \Cutvfs{\posv{(\lb xk.M)}}{\garg{\posv{W}}{y}{\apos{N}}}\\
		& \red_{\Bv} & \Cutvfs{\posv{W}}{x.\Cutvfsc{\apos{M}}{y.\apos{N}}}\\
		& = & \Cutvfs{\posv{W}}{x.\apos{(\sub{\lb y.N}{k}{M})}} & \text{(by Lemma \ref{lem:cong-col-CPS})}\\
		& = &  \apos{((\lb x. \sub{\lb y.N}{k}{M})W)}
	\end{array}
	$$
	
\end{proof}

\subsection{Proofs of Section \ref{sec:bds}}

Next we are going to prove Theorem \ref{thm:two-isomorphisms}. Consider the translations given in Tables \ref{table:ves-vfs} and  \ref{table:ces-cnf}. 
The next 8 Lemmas are needed to obtain Theorem \ref{thm:two-isomorphisms}.

\begin{lem}\label{lem:psi-subs}
	For all $V,W, c_y, M \in \VES$: 
\begin{enumerate}
	\item $\Psi_v (\sub{V}{x}{W}) =  \sub{\Psi_v V}{x}{\Psi_v W}$. 
	\item $\Psi_y (\sub{V}{x}{c_y}) =  \sub{\Psi _vV}{x}{\Psi_y (c_y)}$. 
	\item $\Psi (\sub{V}{x}{M}) =  \sub{\Psi_vV}{x}{\Psi M}$. 
\end{enumerate}
\end{lem}

\begin{proof}
	Straightforward by simultaneous induction on $W$, $c_y$, and $M$.
\end{proof}

\begin{lem}\label{lem:psi-col}
	For all $c_x, P, M \in \VES$:
	\begin{enumerate}
		\item $\Psi_x(\ltc{z}{c_x}{P}) = (\Psi_x(c_x): z.\Psi P)$.
		\item $\Psi(\ltc{z}{M}{P}) = \Cutvfsc{\Psi M}{z. \Psi P}$.
	\end{enumerate}
\end{lem}

\begin{proof}
	By simultaneous induction on $c_x$ and $M$. We proceed by cases.
	
	Case $c_x = M$.
	$$
	\begin{array}{rcll}
		\Psi_x(\ltc{z}{M}{P}) & = & x.\Psi(\ltc{z}{M}{P})\\
		& = & x. \Cutvfsc{\Psi M}{z. \Psi P} & \text{(by IH)}\\
		& = & ((x.\Psi M):z.\Psi P)\\
		& = & (\Psi_x (M):z.\Psi P)
	\end{array}
	$$
	
	Case $c_x = \lt{y}{xW}{N}$.
	$$
	\begin{array}{rcll}
		&& \Psi_x(\ltc{z}{\lt{y}{xW}{N}}{P}) \\
		& = & \Psi_x(\lt{y}{xW}{\ltc{z}{N}{P}}) \\
		& = & \garg{\Psi W}{y}{\Psi(\ltc{z}{N}{P})}\\
		& = & \garg{\Psi W}{y}{\Cutvfsc{\Psi N}{z. \Psi P}} & \text{(by IH)}\\
		& = & ( \garg{\Psi W}{y}{\Psi N}: z. \Psi P)\\
		& = & (\Psi_x(\lt{y}{xW}{N}): z. \Psi P)
	\end{array}
	$$

	Case $M = V$.
	$$
	\begin{array}{rcll}
		\Psi(\ltc{z}{V}{P}) & = & \Psi(\lt{z}{V}{P})\\
		& = & \Cutvfs{\Psi_v V}{z. \Psi P}\\
		& = & \Cutvfsc{\rt{\Psi_v V}}{z. \Psi P} \\
		& = & \Cutvfsc{\Psi V}{z. \Psi P} 
	\end{array}
	$$
	
	Case $M= \lt{x}{V}{c_x}$.
	$$
	\begin{array}{rcll}
		\Psi(\ltc{z}{\lt{x}{V}{c_x}}{P}) & = & \Psi(\lt{x}{V}{\ltc{z}{c_x}{P}})\\
		& = & \Cutvfs{\Psi V}{\Psi_x(\ltc{z}{c_x}{P})}\\
		& = & \Cutvfs{\Psi V}{(\Psi_x(c_x):z. \Psi P)} & \text{(by IH)}\\
		& = & \Cutvfsc{\Cutvfs{\Psi V}{\Psi_x(c_x)}}{z.\Psi P}\\
		& =& \Cutvfsc{\Psi(\lt{x}{V}{c_x})}{z. \Psi P}
	\end{array}
	$$
	
\end{proof}

\begin{lem}\label{lem:theta-subs}
	For all $V,W, c, M \in \VFS$: 
	\begin{enumerate}
		\item $\Theta (\sub{V}{x}{W}) =  \sub{\Theta V}{x}{\Theta W}$. 
		\item $\Theta_y (\sub{V}{x}{c}) =  \sub{\Theta V}{x}{\Theta_y (c)}$. 
		\item $\Theta (\sub{V}{x}{M}) =  \sub{\Theta V}{x}{\Theta M}$. 
	\end{enumerate}
\end{lem}

\begin{proof}
	Straightforward by simultaneous induction on $W$, $c$ and $M$.
\end{proof}

\begin{lem}\label{lem:phi-col}
	For all $c, N, M \in \VFS$:
	\begin{enumerate}
		\item $\Theta_x((c: y.N)) = \ltc{y}{\Theta_x (c)}{\Theta N}$.
		\item $\Theta(\Cutvfsc{M}{y.N}) = \ltc{y}{\Theta M}{\Theta N}$.
	\end{enumerate}  
\end{lem}

\begin{proof}
	By simultaneous induction on $c$ and $M$.
	
	Case $c=z.M$.
	$$
	\begin{array}{rcll}
		\Theta_x(z.M:y.N) & = & \Theta_x(z.\Cutvfsc{M}{y.N})\\
		& = & \sub{x}{z}{\Theta(\Cutvfsc{M}{y.N})}\\
		& = & \sub{x}{z}{(\ltc{y}{\Theta M}{\Theta N})} & \text{(by IH)}\\
		& = & \ltc{y}{\sub{x}{z}{\Theta M}}{\Theta N}& (z \notin N)\\ 
		& =& \ltc{y}{\Theta_x(z.M)}{\Theta N}
	\end{array}
	$$
	
	Case $c=\garg{W}{z}{P}$.
	$$
	\begin{array}{rcll}
		\Theta_x(\garg{W}{z}{P}:y.N) & = & \Theta_x\garg{W}{z}{\Cutvfsc{P}{y.N}}\\
		& = & \lt{z}{x\Theta_v W}{\Theta_x (\Cutvfsc{P}{y.N})}\\
		& = & \lt{z}{x\Theta_v W}{\ltc{y}{\Theta P}{\Theta N}} & \text{(by IH)}\\
		& = & \ltc{y}{\lt{z}{x\Theta_v W}{\Theta P}}{\Theta N}\\
		& = & \ltc{y}{\Theta_x \garg{W}{z}{P}}{\Theta N}
	\end{array}
	$$
	
	Case $M=\rt{V}$.
	$$
	\begin{array}{rcll}
		\Theta (\Cutvfsc{\rt{V}}{y.N}) &=& \Theta(\Cutvfs{V}{y.N})\\
		&=& \lt{x}{\Theta_v V}{\Theta_x (y.N)}\\
		&=& \lt{x}{\Theta_v V}{\sub{x}{y}{\Theta N}}\\
		&=& \lt{y}{\Theta_v V}{\Theta N}\\
		&=& \ltc{y}{\Theta_v V}{\Theta N}\\
		&=& \ltc{y}{\Theta (\rt{V})}{\Theta N}
	\end{array}
	$$
	
	Case $M=\Cutvfs{V}{c}$.
	$$
	\begin{array}{rcll}
		\Theta (\Cutvfsc{\Cutvfs{V}{c}}{y.N}) &=& \Theta (\Cutvfs{V}{(c:y.N)})\\
		& = & \lt{x}{\Theta_v V}{\Theta_x(c:y.N)}\\
		& = & \lt{x}{\Theta_v V}{\ltc{y}{\Theta_x(c)}{\Theta N}} & \text{(by IH)}\\
		& = & \ltc{y}{\lt{x}{\Theta_v V}{\Theta_x(c)}}{\Theta N}\\
		& = & \ltc{y}{\Theta(\Cutvfs{V}{c})}{\Theta N}
	\end{array}
	$$

\end{proof}

\begin{lem}\label{lem:upsi-sub}
	For all $V,W,M \in \ces$: 
	\begin{enumerate}
		\item $\Upsilon (\sub{V}{x}{W}) =  \sub{\Upsilon V}{x}{\Upsilon W}$. 
		\item $\Upsilon (\sub{V}{x}{M}) =  \sub{\Upsilon V}{x}{\Upsilon M}$. 
	\end{enumerate}
\end{lem}

\begin{proof}
	By an easy simultaneous induction on $W$ and $M$.
\end{proof}

\begin{lem}\label{lem:upsi-LET}
	For all $M,P \in \ces$, $\Upsilon(\ltc{y}{M}{P}) = \lsub{\Upsilon M}{y}{\Upsilon P}$.
\end{lem}

\begin{proof}
	By induction on $M$. 
	
	Case $M = V$.
	$$
	\begin{array}{rcll}
		\Upsilon(\ltc{y}{V}{P}) & = & \Upsilon(\sub{V}{y}{P})\\
		& = & \sub{\Upsilon V}{y}{\Upsilon P} & \text{(by Lemma \ref{lem:upsi-sub})}\\
		& = & \lsub{\Upsilon V}{y}{\Upsilon P}
	\end{array}
	$$
	
	Case $M = \lt{x}{VW}{N}$.
	$$
	\begin{array}{rcll}
		\Upsilon(\ltc{y}{\lt{x}{VW}{N}}{P}) & = & \Upsilon(\lt{x}{VW}{\ltc{y}{N}{P}})\\
		& = & \li{\Upsilon V}{\Upsilon W}{x}{\Upsilon(\ltc{y}{N}{P})} \\
		& = &  \li{\Upsilon V}{\Upsilon W}{x}{\lsub{\Upsilon N}{y}{\Upsilon P}} & \text{(by IH)}\\
		& = & \lsub{\li{\Upsilon V}{\Upsilon W}{x}{\Upsilon N}}{y}{\Upsilon P}\\
		& = & \lsub{\Upsilon(\lt{x}{VW}{N})}{y}{\Upsilon P}
	\end{array}
	$$
	
\end{proof}

\begin{lem}\label{lem:phi-subs}
	For all $V,W,M \in \cnf$: 
	\begin{enumerate}
		\item $\Phi (\sub{V}{x}{W}) =  \sub{\Phi V}{x}{\Phi W}$. 
		\item $\Phi (\sub{V}{x}{M}) =  \sub{\Phi V}{x}{\Phi M}$. 
	\end{enumerate}
\end{lem}

\begin{proof}
	By an easy simultaneous induction on $W$ and $M$.
\end{proof}

\begin{lem}\label{lem:phi-LET}
	For all $M,P \in \cnf$, $\Phi (\lsub{M}{x}{P}) = \ltc{x}{\Phi M}{\Phi P}$.
\end{lem}

\begin{proof}
	By induction on $M$. 
	
	Case $M = V$.
	$$
	\begin{array}{rcll}
		\Phi(\lsub{V}{x}{P}) & = & \Phi(\sub{V}{x}{P})\\
		& = & \sub{\Phi V}{x}{\Phi P} & \text{(by Lemma \ref{lem:phi-subs})}\\
		& = & \ltc{x}{\Phi V}{\Phi P}
	\end{array}
	$$
	
	Case $M = \li{V}{W}{y}{N}$.
	$$
	\begin{array}{rcll}
		\Phi(\lsub{\li{V}{W}{y}{N}}{x}{P}) & = & \Phi (\li{V}{W}{y}{\lsub{N}{x}{P}})\\
		& = & \lt{y}{\Phi V \Phi W}{\Phi(\lsub{N}{x}{P})}\\
		& = & \lt{y}{\Phi V \Phi W}{\ltc{x}{\Phi N}{\Phi P}} & \text{(by IH)}\\
		& = & \ltc{x}{\lt{y}{\Phi V \Phi W}{\Phi N}}{\Phi P}\\
		& = & \ltc{x}{\Phi(\li{V}{W}{y}{N})}{\Phi P}
	\end{array}
	$$
	
\end{proof}

\noindent \textbf{Theorem \ref{thm:two-isomorphisms}}\textbf{.}
$\VES\cong\VFS$ and $\ces\cong\cnf$.

\begin{proof}
	First we are going to prove the isomorphism $\VES\cong\VFS$. The following items need to be proven: 
	\begin{enumerate}
		\item For all $M,V, c_x\in\VES$,  $\Theta(\Psi M)=M$, $\Theta(\Psi V)=V$ and $\Theta_x(\Psi_x (c_x)) = c_x$.
		\item For all $M,V,c\in\VFS$, $\Psi(\Theta M)=M$, $\Psi(\Theta V)=V$ and $\Psi_x(\Theta_x(c))=c$.
		\item If $M_1\to M_2$ in $\VES$ then $\Psi M_1 \to \Psi M_2$ in $\VFS$.
		\item If $M_1\to M_2$ in $\VFS$ then $\Theta M_1 \to \Theta M_2$ in $\VES$.
	\end{enumerate}

	Item 1 follows by simultaneous induction on $M$, $V$ and $c_x$. Similarly, item 2 follows by simultaneous induction on $M$, $V$ and $c$. Item 3 follows by induction on the relation $M_1 \to M_2$ in $\VES$. We just show the base cases.
	
	Case $(\ltv)$.
	$$
	\begin{array}{rcll}
		\Psi(\lt{x}{V}{M}) & = & \Cutvfs{\Psi V}{x. \Psi M}\\
		& \red_{\sigmav} & \sub{\Psi V}{x}{\Psi M} \\
		& = & \Psi (\sub{V}{x}{M})& \text{(by Lemma \ref{lem:theta-subs})}
	\end{array}
	$$
	
	Case $(\Bv)$.
	$$
	\begin{array}{rcll}
		&& \Psi(\lt{y}{\lb x.M}{\lt{z}{yV}{P}}) \\
		& = & \Cutvfs{\lb x. \Psi M}{\garg{\Psi V}{z}{\Psi P}}\\
		& \red_{\Bv} &  \Cutvfs{\lb x. \Psi M}{x. \Cutvfsc{\Psi M}{z. \Psi P}}\\
		& = & \Cutvfs{\Psi V}{x. \Psi(\ltc{z}{M}{P})}& \text{(by Lemma \ref{lem:psi-col})}\\
		& = & \Psi(\lt{x}{V}{\ltc{z}{M}{P}})
	\end{array}
	$$
	
		Item 4 follows by induction on the relation $M_1 \to M_2$ in $\VFS$. We just show the base cases.
		
		Case $(\sigmav)$.
		$$
		\begin{array}{rcll}
			\Theta (\Cutvfs{V}{y.N})& = & \lt{y}{\Theta V}{\Theta N}\\
			& \red_{\ltv} & \sub{\Theta V}{x}{\Theta N} \\
			& = & \Theta (\sub{V}{x}{N})& \text{(by Lemma \ref{lem:phi-subs})}
		\end{array}
		$$
		
		Case $(\Bv)$.
		$$
		\begin{array}{rcll}
			&& \Theta (\Cutvfs{\lb x.M}{\garg{V}{y}{N}}) \\
			& = & \lt{z}{\lb x. \Theta M}{\lt{y}{z\Theta V}{\Theta N}}\\
			& \red_{\Bv} & \lt{x}{\Theta V}{\ltc{y}{\Theta M}{\Theta N}}\\
			& =& \lt{x}{\Theta V}{\Theta(\Cutvfsc{M}{y.N})} & \text{(by Lemma \ref{lem:phi-col})}\\
			& =& \Theta (\Cutvfs{V}{x.\Cutvfsc{M}{y.N}})
		\end{array}
		$$

		Now we are going to prove the second isomorphism $\ces\cong\cnf$. The following items need to be proven: 
		
		\begin{enumerate}
			\item For all $M,V \in \ces$,  $\Phi(\Upsilon M)=M$ and $\Phi(\Upsilon V)=V$.
			\item For all $M,V \in\cnf$, $\Upsilon(\Phi M)=M$ and $\Upsilon(\Phi V)=V$.
			\item If $M_1\to M_2$ in $\ces$ then $\Upsilon M_1 \to \Upsilon M_2$ in $\cnf$.
			\item If $M_1\to M_2$ in $\cnf$ then $\Phi M_1 \to \Phi M_2$ in $\ces$.
		\end{enumerate}
		
		Item 1 follows by simultaneous induction on $M$ and $V$. Analogously, item 2 follows by simultaneous induction on $M$ and $V$.  Item 3 follows by induction on the relation $M_1 \to M_2$ in $\ces$. We just show the base case.
		
		Case $(\betav)$.
		$$
		\begin{array}{rcll}
			\Upsilon(\lt{y}{(\lb x.M)V}{P}) & = & \li{\Upsilon (\lb x.M)}{\Upsilon V}{y}{\Upsilon P}\\
			& = & \li{\lb x. \Upsilon M}{\Upsilon V}{y}{\Upsilon P}\\
			& \red_{\betav} & \lsub{\sub{\Upsilon V}{x}{\Upsilon M}}{y}{\Upsilon P}\\
			& = & \lsub{\Upsilon(\sub{V}{x}{M})}{y}{\Upsilon P} & \text{(by Lemma \ref{lem:upsi-sub})}\\
			& = & \Upsilon(\ltc{y}{\sub{V}{x}{M}}{P}) & \text{(by Lemma \ref{lem:upsi-LET})}
		\end{array}
		$$
		
		Similarly, item 4 follows by induction on the relation $M_1 \to M_2$ in $\cnf$. We just show the base case.
		
		Case $(\betav)$.
		$$
		\begin{array}{rcll}
			\Phi(\li{(\lb y.M)}{W}{x}{P}) & = & \lt{x}{\Phi(\lb y. M)\Phi W}{\Phi P}\\
			& = & \lt{x}{(\lb y. \Phi M)\Phi W}{\Phi P} \\
			& \red_{\betav} & \ltc{x}{\sub{\Phi W}{y}{\Phi M}}{\Phi P}\\
			& = & \ltc{x}{\Phi(\sub{W}{y}{M})}{\Phi P} & \text{(by Lemma \ref{lem:phi-subs})}\\
			& = & \Phi(\lsub{\sub{W}{y}{M}}{x}{P}) & \text{(by Lemma \ref{lem:phi-LET})}
		\end{array}
		$$
\end{proof}

Next we are going to prove Theorem \ref{thm:iso-cnf}. 

\begin{table}[t]
		$$
		\begin{array}{rclcrcl}
			\ngv x & = & x &\qquad&  \angt {V} & = & k\ngv V\\
			\ngv{(\lb x.M)} & = & \lb x. \ngt{M} &&\angt{(V{\garg{W}{x}{M}})}& = & \ngv V \ngv W (\lb x.\angt M) \\
			\ngt M & = & \lb k.\angt M && \\
			\\
			\posv{x}&=&x&&\apos{(kV)}&=&\posv{V}\\
			\posv{(\lb x.P)} &=&\lb x.\post{P}&&\apos{(VW(\lb x.M))}&=&\posv{V}\garg{\posv{W}}{x}{\post{M}}\\
			\post{(\lb k.M)}&=&\apos M&&
		\end{array}
		$$
	\caption{Translation from $\cnf$ to $\cps$ and vice-versa.}
	\label{table:cnf-cps}
\end{table}

Consider the translations given in table \ref{table:cnf-cps}.  The next 4 Lemmas are needed to obtain Theorem \ref{thm:iso-cnf}.

\begin{lem}\label{lem:subs-neg}
	For all $V,W,M \in \cnf$:
	\begin{enumerate}
		\item $\ngv{(\sub{V}{x}{W})} = \sub{\ngv{V}}{x}{\ngv{W}}$.
		\item $\angt{(\sub{V}{x}{M})} = \sub{\ngv{V}}{x}{\angt{M}}$.
	\end{enumerate}
\end{lem}

\begin{proof}
	By simultaneous induction on $W$ and $M$.
\end{proof}

\begin{lem}\label{lem:col-neg}
	For all $M,N \in \cnf$, $\angt{(\lsub{N}{x}{M})} = \sub{\lb x. \angt{M}}{k}{\angt{N}}$.
\end{lem}

\begin{proof}
	By induction on $N$.
	
	Case $N=V$.
	$$
	\begin{array}{rclcrcl}
		\angt{(\lsub{V}{x}{M})} &=& \angt{(\sub{V}{x}{M})}\\
		& =& \sub{\ngv{V}}{x}{\angt{M}} & \text{(by Lemma \ref{lem:subs-neg})}\\
		&=& \sub{\lb x. \angt{M}}{k}{(k\ngv{V})} \\
		&=&  \sub{\lb x. \angt{M}}{k}{\angt{V}} 
	\end{array}
	$$
	
	Case $N=\li{V}{W}{y}{N}$.
	$$
	\begin{array}{rclcrcl}
		\angt{(\lsub{\li{V}{W}{y}{N}}{x}{M})} &=& \angt{(\li{V}{W}{y}{\lsub{N}{x}{M}})}\\
		&=& \ngv{V}\ngv{W}(\lb y. \angt{(\lsub{N}{x}{M})}) \\
		&=& \ngv{V}\ngv{W}(\lb y. \sub{\lb x.\angt{M}}{k}{\angt{N}}) & \text{(by IH)}\\
		&=& \sub{\lb x. \angt{M}}{k}{(\ngv{V}\ngv{W}(\lb y. \angt{N}))}\\
		&=& \sub{\lb x. \angt{M}}{k}{\angt{(\li{V}{W}{y}{N})}}
	\end{array}
	$$
\end{proof}

\begin{lem}\label{lem:subs-cps}
	For all $V,W,M,P \in \cps$:
	\begin{enumerate}
		\item $\posv{(\sub{V}{x}{W})} = \sub{\posv{V}}{x}{\posv{W}}$.
		\item $\apos{(\sub{V}{x}{M})} = \sub{\posv{V}}{x}{\apos{M}}$.
		\item $\post{(\sub{V}{x}{P})} = \sub{\posv{V}}{x}{\post{P}}$.
	\end{enumerate} 
\end{lem}

\begin{proof}
	By simultaneous induction on $W, M$ and $P$.
\end{proof}

\begin{lem}\label{lem:col-cps}
	For all $M,N \in \cps$, $\apos{(\sub{\lb x.N}{k}{M})} = \lsub{\apos{M}}{x}{\apos{N}}$.
\end{lem}

\begin{proof}
	By induction on $M$.
	Case $M=kV$.
	$$
	\begin{array}{rclcrcl}
		\apos{(\sub{\lb x.N}{k}{(kV)})} &=& \apos{(\sub{V}{x}{N})}\\
		&=& \sub{\posv{V}}{x}{\apos{N}} & \text{(by Lemma \ref{lem:subs-cps})}\\
		&=& \lsub{\posv{V}}{x}{\apos{N}}\\
		&=& \lsub{\apos{(kV)}}{x}{\apos{N}}
	\end{array}
	$$
	
	Case $M=VW(\lb y.P)$.
	$$
	\begin{array}{rclcrcl}
		\apos{(\sub{\lb x.N}{k}{(VW(\lb y.P))})} &=& \apos{VW(\lb y.\sub{\lb x.N}{k}{P})}\\
		&=& \li{\posv{V}}{\posv{W}}{y}{\apos{(\sub{\lb x.N}{k}{P})}}\\
		&=& \li{\posv{V}}{\posv{W}}{y}{\lsub{\apos{P}}{x}{\apos{N}}} & \text{(by IH)}\\
		&=& \lsub{\li{\posv{V}}{\posv{W}}{y}{\apos{P}}}{x}{\apos{N}}\\
		&=& \lsub{\apos{(VW(\lb y.P))}}{x}{\apos{N}}
	\end{array}
	$$
	
\end{proof}

\noindent \textbf{Theorem \ref{thm:iso-cnf}}\textbf{.}
$\cnf \cong \cps$.
\begin{enumerate}
	\item For all $M,V\in\cnf$, $\post{{\ngt{M}}}=M$ and $\apos{{\angt{M}}}=M$ and $\posv{{\ngv{V}}}=V$.
	\item For all $P,M,V\in\cps$, $\ngt{{\post{P}}}=P$ and $\angt{{\apos{M}}}=M$ and $\ngv{{\posv{V}}}=V$.
	\item If $M_1\to M_2$ in $\cnf$ then $\angt{M_1}\to\angt{M_2}$ in $\cps$ (hence $\ngt{M_1}\to\ngt{M_2}$ in $\cps$).
	\item If $M_1\to M_2$ in $\cps$ then $\apos{M_1}\to\apos{M_2}$ in $\cnf$. Hence If $P_1\to P_2$ in $\cps$ then $\post{P_1}\to\post{P_2}$ in $\cnf$.
\end{enumerate}

\begin{proof}
	Item 1 is straightforward by simultaneous induction on $M$ and $V$. Similarly, item 2 is straightforward by simultaneous induction on $P, M$ and $V$. Item 3 follows by induction on the relation $M_1 \to M_2$ in $\cnf$. We just show the base case.
	
	Case $(\betav)$.
	$$
	\begin{array}{rcll}
		\angt{((\lb y.M)\garg{V}{x}{P})} & = & (\lb y. \lb k. \angt{M})\ngv{V}(\lb x.\angt{P})\\
		& \red_{\betav} & \sub{\lb x. \angt{P}}{k}{\sub{\ngv{V}}{y}{\angt{M}}} \\
		&=&  \sub{\lb x. \angt{P}}{k}{\angt{(\sub{V}{y}{M})}} & \text{(by Lemma \ref{lem:subs-neg})}\\
		&=&  \angt{(\lsub{\sub{V}{y}{M}}{x}{P})} &\text{(by Lemma \ref{lem:col-neg})}
	\end{array}
	$$
	
	Analogously, item 4 follows by induction on the relation $M_1 \to M_2$ in $\cps$. We just show the base case.
	
	Case $(\betav)$.
	$$
	\begin{array}{rcll}
		\apos{((\lb y. \lb x.M)W(\lb x.N))} & = & (\lb y. \apos{M})\garg{\posv{W}}{x}{\post{N}}\\
		& \red_{\betav} & \lsub{\sub{\posv{W}}{y}{\apos{M}}}{x}{\post{N}}\\
		& = & \lsub{\apos{(\sub{W}{y}{M})}}{x}{\post{N}} & \text{(by Lemma \ref{lem:subs-cps})}\\
		& = & \apos{(\sub{\lb x.N}{k}{\sub{W}{y}{M}})} & \text{(by Lemma \ref{lem:col-cps})}
	\end{array}
	$$

\end{proof}